\documentclass[12pt,a4paper,dvipdfmx]{amsart}
\usepackage{amsmath}
\usepackage{amsthm}
\usepackage{amsfonts}
\usepackage{amssymb}
\usepackage{latexsym}
\usepackage{amscd}
\usepackage{stmaryrd}
\usepackage{amsbsy}
\usepackage{pb-diagram}
\usepackage{tikz}
\allowdisplaybreaks
\usepackage{txfonts}

\newcommand{\Vir}{\mathrm{Vir}}
\newcommand{\Z}{\mathbb{Z}}
\newcommand{\C}{\mathbb{C}}
\newcommand{\Lam}{|\Lambda\rangle}

\theoremstyle{plain}
\newtheorem{thm}{Theorem}[section]
\newtheorem{cor}[thm]{Corollary}
\newtheorem{lem}[thm]{Lemma}
\newtheorem{prop}[thm]{Proposition}
\newtheorem{conj}[thm]{Conjecture}
\newtheorem{dfn}[thm]{Definition}

\newtheorem{re}[thm]{Remark}

\numberwithin{equation}{section}

\begin{document}

\title{Irregular conformal blocks, with an application to the fifth and fourth Painlev\'e equations}
\author{Hajime Nagoya}
\address{Rikkyo University}
\email{nagoya.hajime@rikkyo.ac.jp}
\begin{abstract}
We develop the theory of irregular conformal 
blocks of the Virasoro algebra. 
In previous studies,  expansions of 
irregular conformal 
blocks at regular singular points were 
obtained as degeneration limits of regular conformal blocks; however,  
such expansions  
at irregular singular points were not 
clearly understood. This is because 
precise definitions of irregular vertex 
operators had not been provided previously. 
In this paper, 
we present precise definitions of 
irregular vertex operators of two types 
and we prove that one of our vertex operators 
 exists uniquely. 
Then, we define irregular 
conformal blocks with at most two 
irregular singular points as 
expectation values of given 
irregular vertex operators.  
Our definitions  
provide an understanding of 
 expansions of irregular conformal blocks 
and enable us to obtain  
expansions  
at irregular singular points.

As an application, we propose conjectural formulas 
of series expansions of the tau functions of the fifth and fourth Painlev\'e equations, 
using expansions  of irregular conformal blocks  at an irregular singular point.  
\end{abstract}
\maketitle

\section{Introduction}


Conformal blocks  are 
building blocks of the correlation functions 
of two-dimensional 
conformal field theory \cite{BPZ}.  
They are defined as the expectation values 
\begin{equation*}
\left\langle
\Phi(z_N)\cdots\Phi(z_1)
\right\rangle
\end{equation*}
of the vertex operators $\Phi(z_i)$
on the Verma 
modules of the Virasoro algebra satisfying 
the commutation relations
\begin{equation*}
\left[L_n, \Phi(z_i)\right]=z_i^n\left(z_i\frac{\partial}{\partial z_i}
+(n+1)\Delta_i \right)\Phi(z_i),  
\end{equation*} 
where $L_n$ ($n\in \Z$) are generators of the  Virasoro algebra, and the complex parameters 
$\Delta_i$ ($i=1,\ldots,N$) are conformal dimensions. The conformal blocks 
can be viewed as special functions 
with Virasoro symmetry.

Conformal blocks can be calculated 
from the definition. They are formal power 
series in $z_i/z_{i+1}$ ($i=1,\ldots, N-1$) 
and are believed to be absolutely convergent 
in $|z_1|<\cdots <|z_N|$.  
As a special case, the 4-point conformal block with one null vector condition 
becomes the hypergeometric series. 
In the general case, 
however, their 
explicit series expansions had not been 
determined.  
Recently, 
Alday, Gaiotto and Tachikawa proposed that 
the Virasoro conformal 
blocks correspond to the instanton parts of 
the Nekrasov partition functions of a particular class of four-dimensional  supersymmetric 
gauge theories \cite{AGT}, and 
their proposed correspondence was proved in \cite{AFLT}. 
 Given this so-called AGT correspondence, we have  
explicit series expansions of 
the conformal blocks. 

Conformal blocks with null vector 
conditions satisfy 
partial differential systems 
with regular singularities, so-called
Belavin-Polyakov-Zamolodchikov equations (BPZ), whereas 
the general conformal blocks whose 
central charge $c=1$ appear in series 
expansions of the tau functions of 
 Garnier systems \cite{Iorgov Lisovyy Teschner}, which describe isomonodromy deformations of $2\times 2$ 
Fuchsian systems. Therefore,  
conformal blocks can be viewed 
as special functions with regular 
singularities.

Irregular versions of conformal blocks have 
been studied in relation to four-dimensional 
supersymmetric gauge theories and 
quantum Painlev\'e equations. 
Note that canonical quantization 
of the sixth Painlev\'e equation 
is the BPZ equation for five points. 
Particular 
irregular vertex operators of Wess-Zumino-Novikov-Witten conformal 
field theory were presented by free field realizations,  and  
a recursion rule for constructing integral representations of 
irregular  conformal blocks  was provided 
in \cite{Nagoya Sun 2010}. 
These realizations 
are also applicable to free field realizations of irregular vertex operators 
and irregular conformal blocks in 
 Virasoro conformal field theory. 
Hence, in a special case, we have integral formulas of irregular conformal blocks.

Another method of obtaining irregular conformal blocks, initiated by Gaiotto \cite{Gaiotto}, 
is to take pairings of 
irregular vectors $|\Lambda\rangle$ and $\langle \Lambda' |$ 
embedded in a Verma module and a dual Verma module, respectively, 
such that for non-negative integers $r,s,$ and tuples $\Lambda=(\Lambda_r,\Lambda_{r+1},\ldots, 
\Lambda_{2r})$, $\Lambda'=(\Lambda'_s,\Lambda'_{s+1},\ldots, 
\Lambda'_{2s})$, 
\begin{align}
&L_n |\Lambda\rangle=\Lambda_n |\Lambda\rangle \quad (n=r,\ldots, 2r),
\quad L_n |\Lambda\rangle=0\quad (n>2r), \label{eq_irregular_vector}
\\
&\langle \Lambda' |L_n=\Lambda'_n\langle \Lambda' |\quad (n=-s,\ldots, -2s),
\quad \langle \Lambda' |L_n=0 \quad (n<-2s).  \nonumber
\end{align} 
Another type of irregular vector was 
also introduced in \cite{BMT}. 
The existence of irregular vectors in a Verma module of the Virasoro algebra 
of the $r=1$ case was verified in \cite{MMM}, \cite{Yanagida}. For the $r>1$ case, 
the construction of irregular vectors was performed in \cite{FJK} and it was 
 revealed that the first-order irregular vectors are uniquely determined 
by the condition \eqref{eq_irregular_vector}, up to a scalar.  However, the higher-order irregular vectors contain 
an infinite number of parameters, and thus they are not unique.  

On the other hand, as the confluent hypergeometric series is obtained from 
the hypergeometric series by confluence, 
irregular conformal blocks can be derived from ordinary conformal blocks 
using particular delicate
 limiting procedures \cite{AFKMY 2010}, 
 \cite{GT}, \cite{Nishinaka Rim 2012}, \cite{Choi Rim 2015}. 
 We emphasize  that some properties of irregular conformal blocks might be 
 limits of regular conformal blocks; however, it might be easier to derive the results 
 directly than to justify the limiting procedures, as is true for the hypergeometric function 
 and  the confluent hypergeometric function.



In this study,  we provide precise  
definitions of irregular vertex 
operators of two types  
and we prove that one of our vertex operators 
exists uniquely. 
Then, we directly 
define irregular conformal blocks with at most two 
irregular singular points, as 
expectation values of the newly introduced 
vertex 
operators. 
Recall that in the regular case, the vertex 
operators are uniquely determined from the 
commutation relations and three conformal dimensions, 
and 
the conformal blocks are defined 
as expectation values of the vertex 
operators.  
 
Our definitions  
provide an understanding of 
series expansions of irregular conformal blocks 
and  enable us to obtain series 
expansions of irregular conformal blocks 
at irregular singular points.

As an application, we present conjectural formulas 
of series expansions of the tau functions of the fifth and fourth Painleve equations, 
using our irregular conformal blocks, 
following a recent remarkable discovery 
by Gamayun, Iorgov and Lisovyy \cite{GIL1} that 
the tau function of the sixth Painlev\'e 
equation is expressed as 
a series expansion in terms of 
a 4-point conformal block. By taking 
scaling limits, short-distance expansions 
of the tau functions of $\mathrm{P_V}$, 
$\mathrm{P_{III}}$ in terms of 
irregular conformal blocks 
were obtained in \cite{GIL2}. Furthermore, 
it was conjectured in 
\cite{Its Lisovyy Tykhyy 2014} that 
a long-distance expansion of 
the tau function of a special case of 
the third Painlev\'e equation $\mathrm{P_{III}}(D_8)$ \cite{OKSO} is 
also a series expansion in terms of some function  
 determined from the differential equation
 $\mathrm{P_{III}}(D_8)$.

 In the $\mathrm{P_V}$ case, the tau function 
is expected to be a series expansion in terms  
of 3-point irregular conformal 
blocks  with two regular singular points and one irregular 
singular point of rank 1. Naively, this is expressed 
as 
\begin{equation*}
\langle \Delta | \Phi(z) |\Lambda\rangle,
\end{equation*}
where $\langle \Delta |$ is the 
highest weight vector of the 
dual Verma module $V^*_{\Delta}$, 
$\Phi(z)$ is a usual vertex operator,  
and $|\Lambda\rangle$ is 
the  irregular vector of rank 1 with $\Lambda=(\Lambda_1,\Lambda_2)$. 
Our definitions of irregular conformal blocks distinguish between 
\begin{equation}\label{PV regular}
\left( \langle \Delta | \Phi^{*,\Delta_2}_{\Delta,\Delta_1}(z)\right)\cdot  |\Lambda\rangle, 
\end{equation}
and 
\begin{equation}\label{PV irregular}
\langle \Delta| \cdot \left(\Phi^{\Delta_2}_{\Lambda', \Lambda}(z) |\Lambda\rangle\right), 
\end{equation}
where $\Phi^{*,\Delta_2}_{\Delta,\Delta_1}(z)$: $V^*_\Delta\to V^*_{\Delta_1}$ 
is the usual vertex operator,  $\Phi^{\Delta_2}_{\Lambda', \Lambda}(z)$: $V^{[1]}_{\Lambda}
\to V^{[1]}_{\Lambda'}$ is 
the newly introduced rank $0$ vertex operator from a rank $1$ Verma module 
to another rank $1$ Verma module such that $\Lambda'=(\Lambda_1',\Lambda_2)$. 
The expectation values \eqref{PV regular}, \eqref{PV irregular} are computed by the pairings 
$V^*_{\Delta_1}\times V^{[1]}_\Lambda\to\C$, 
$V^*_{\Delta}\times V^{[1]}_{\Lambda'}\to\C$, respectively. 
The irregular conformal block \eqref{PV regular} is the expansion at a regular singular point 
and, in fact,  is equal to the one  
derived by the collision limit of the 4-point conformal block and 
the building block of 
the short-distance expansion of the fifth Painlev\'e tau function obtained 
in \cite{GIL2}. Using the irregular conformal block \eqref{PV irregular},  
the expansion at an irregular singular point, we present   
a conjectural long-distance expansion of the fifth Painlev\'e tau function. 

In the $\mathrm{P_{IV}}$ case, as a building block of a series expansion of the tau function,
 we use the 2-point irregular conformal 
blocks  with one regular singular point and one irregular 
singular point of rank 2. This is expressed as 
follows: 
\begin{equation*}
\langle 0| \cdot \left( \Phi^{\Delta_2}_{\Lambda',\Lambda}(z) |\Lambda\rangle\right), 
\end{equation*}
where $\Phi^{\Delta_2}_{\Lambda',\Lambda}(z)$: $V^{[2]}_{\Lambda}
\to V^{[2]}_{\Lambda'}$ is 
the newly introduced rank $0$ vertex operator from a rank $2$ Verma module 
to another rank $2$ Verma module, $\Lambda=(\Lambda_2,\Lambda_3,\Lambda_4)$ 
and $\Lambda'=(\Lambda_2',\Lambda_3,\Lambda_4)$. 
The expectation value is computed by the pairing  
$V^*_{0}\times V^{[2]}_{\Lambda'}\to\C$. 


The remainder of this paper is organized as follows. 
In Section 2, we introduce a rank $r$ Verma module of 
the Virasoro algebra as a special Whittaker module of 
that algebra. Then, generalizing  
the confluent primary field introduced in \cite{Nagoya Sun 2010}, 
we present definitions of irregular vertex operators of two types, 
and we prove that one of our vertex operators 
 exists uniquely. 
In Section 3, we define irregular conformal blocks 
with at most two 
irregular singular points, using the newly introduced irregular vertex operators 
in the previous section. 
In Section 4, we review  the series expansions of the tau functions 
of the Painlev\'e equations. Then, we propose 
series expansions of the tau functions of $\mathrm{P_V}$ and $\mathrm{P_{IV}}$ 
in terms of our irregular conformal blocks. 

\section{Irregular vertex operators}

\subsection{Modules}

The Virasoro algebra 
\begin{equation*}
\Vir=\bigoplus_{n\in\Z}\C L_n\oplus \C C
\end{equation*}
is the Lie algebra satisfying the following commutation relations:
\begin{align*}
&[L_m, L_n]=(m-n)L_{m+n}+\frac{1}{12}(m^3-m)\delta_{m+n,0}C,
\\
&[\Vir, C]=0,
\end{align*}
where $\delta_{i,j}$ stands for Kronecker's delta. 

Let $r$ be a non-negative integer and $\Vir_r$ the subalgebra generated by $C$ and 
$L_n$ ($n\in\Z_{\ge r}$). 
For a complex number $c$ and an ($r+1$)-tuple of parameters 
$\Lambda=(\Lambda_r,\Lambda_{r+1},\ldots,\Lambda_{2r})\in \C^{r}\times \C^\times$, 
let $\C \Lam$ be the one-dimensional $\Vir_r$-module defined by 
\begin{align*}
&C\Lam=c\Lam, 
\\
&L_n\Lam=\Lambda_n\Lam \quad (n=r,r+1,\ldots, 2r),
\\
&L_n\Lam=0\quad (n>2r). 
\end{align*} 

\begin{dfn}
A rank $r$ Verma module $V^{[r]}_{\Lambda}$ with parameters $\Lambda$ is the induced module 
\begin{equation*}
V^{[r]}_{\Lambda}=\mathrm{Ind}_{\Vir_r}^\Vir \C \Lam\ (=U(\Vir)\otimes_{\Vir_r}\C \Lam). 
\end{equation*}

\end{dfn}
 
 Denote the dual module of $V^{[r]}_{\Lambda}$ by $V^{*,[r]}_{\Lambda}$ such that  
\begin{align*}
&\langle \Lambda |C=c\langle \Lambda |, 
\\
&\langle \Lambda |L_n=\Lambda_n\langle \Lambda | \quad (n=-r,-r-1,\ldots, -2r),
\\
&\langle\Lambda|L_n=0\quad (n<-2r). 
\end{align*} 
  The module $V^{*,[r]}_{\Lambda}$ is spanned by linearly independent vectors 
  of the form
\begin{equation*}
\langle \Lambda |L_{i_1}\cdots L_{i_k}\quad (-r<i_1\le \cdots\le i_k). 
\end{equation*}

\begin{re}
When $r=0$, a rank $0$ Verma module $V^{[0]}_\Lambda$ is a standard Verma module of 
the Virasoro algebra. It is usual to denote $V^{[0]}_\Lambda$ as $V_\Delta$ in  conformal 
field theory. In the case of $r\ge 1$, a rank $r$ Verma module is equal to the 
universal Whittaker module discussed in \cite{OW},  \cite{LGZ} and \cite{FJK}.   
It has been shown that the universal Whittaker modules are irreducible,  if $\Lambda_{2r}\neq 0$ or 
$\Lambda_{2r-1}\neq 0$  \cite{LGZ} 
and \cite{FJK}. Hence, a rank $r$ Verma module $V^{[r]}_{\Lambda}$ 
for $r\ge 1$  
is irreducible. 
\end{re}

  A bilinear pairing $\langle \cdot \rangle$: $ V^{*,[0]}_\Delta
  \times V^{[1]}_\Lambda\to \C$ is 
uniquely defined by 
\begin{align*}
&\langle \Delta| \cdot |\Lambda\rangle=1,
\\
&\langle u|L_n \cdot |v\rangle =\langle u| \cdot L_n |v\rangle \equiv 
\langle u| L_n |v\rangle, 
\end{align*}
where $\langle u|\in V^{*,[0]}_\Delta$ and $|v\rangle\in V^{[1]}_\Lambda$. 
Denote by $\bar{V}^{*,[0]}_0$ an irreducible 
highest weight module.   
A bilinear pairing $\langle \rangle$: $ \bar{V}^{*,[0]}_0
\times V^{[r]}_\Lambda\to \C$ is 
uniquely defined by 
\begin{align*}
&\langle 0| \cdot |\Lambda\rangle=1,
\\
&\langle u|L_n \cdot |v\rangle =\langle u| \cdot L_n |v\rangle \equiv 
\langle u| L_n |v\rangle, 
\end{align*}
where $\langle u|\in \bar{V}^{*,[0]}_0$ 
and $|v\rangle\in V^{[2]}_\Lambda$, because 
$\langle 0 |L_{1}=0$.   
The bilinear pairings 
on $ V^{*,[1]}_\Lambda
  \times V^{[0]}_\Delta\to \C$ 
  and 
  $ V^{*,[2]}_\Lambda
  \times \bar{V}^{[0]}_0\to \C$ 
  are defined in the same manner.

 For cases other than  the above,  it is possible to define a 
 bilinear pairing using the Heisenberg algebra and its Fock spaces. 
 Let  
 \begin{equation*}
 \mathcal{H}=\bigoplus_{n\in\Z}\C a_n,
 \end{equation*}
 be the Lie algebra with commutation relations
 \begin{equation*}
 [a_m,a_n]=m\delta_{m+n,0}. 
 \end{equation*}
 We call this the Heisenberg algebra. 
 \begin{dfn}
  A (bosonic) rank $r$ Fock space of $\mathcal{H}$ is  
  a representation $\mathcal{F}^{[r]}_{\lambda}$ such that for a given tuple of 
  parameters $\lambda=(\lambda_0,\ldots,\lambda_r)\in \C^r\times \C^\times$, 
\begin{equation*}
a_n|\lambda\rangle=\lambda_n|\lambda\rangle\quad (n=0,1,\ldots, r),
\quad a_n|\lambda\rangle=0\quad (n>r)
\end{equation*}  
  and $\mathcal{F}^{[r]}_{\lambda}$ is the linear span of linearly independent 
  vectors of the form
  \begin{equation*}
  a_{-i_k}\cdots a_{-i_1}|\lambda\rangle\quad (0<i_1\le\cdots\le i_k). 
  \end{equation*}
 \end{dfn}
\begin{dfn} 
  A rank $r$ dual Fock space of $\mathcal{H}$ is  
  a representation $\mathcal{F}^{*,[r]}_{\lambda}$ such that for a given tuple of 
  parameters $\lambda=(\lambda_0,\ldots,\lambda_r)\in \C^r\times \C^\times$, 
\begin{equation*}
\langle \lambda|a_n=\lambda_n\langle \lambda|\quad (n=0,-1,\ldots, -r),
\quad \langle \lambda|a_n=0\quad (n<-r)
\end{equation*}  
  and $\mathcal{F}^{*,[r]}_{\lambda}$ is the linear span of linearly independent 
  vectors of the form
  \begin{equation*}
 \langle \lambda| a_{i_1}\cdots a_{i_k}\quad (0<i_1\le\cdots\le i_k). 
  \end{equation*}
 \end{dfn}
 For $\lambda=(\lambda_0,\ldots,\lambda_r)\in \C^r\times \C^\times$ 
 and $\mu=(\lambda_0,\mu_1\ldots,\mu_s)\in \C^s\times \C^\times$, 
  a bilinear pairing $\langle \cdot \rangle$: $ \mathcal{F}^{*,[s]}_\mu
 \times \mathcal{F}^{[r]}_\lambda\to \C$ is 
 defined by 
\begin{align*}
&\langle \mu| \cdot |\lambda\rangle=1,
\\
&\langle u|a_n \cdot |v\rangle =\langle u| \cdot a_n |v\rangle \equiv 
\langle u| a_n |v\rangle, 
\end{align*}
where $\langle u|\in \mathcal{F}^{*,[s]}_\mu$ and $|v\rangle\in \mathcal{F}^{[r]}_\lambda$. 
Note that the first elements of $\lambda$ and $\mu$ must be the same. 
  
  On a rank $r$ (dual) Fock space, we can define the structure of 
  a $\Vir$-module by 
  the following formulas: 
  \begin{align}
  &L_n=\frac{\epsilon}{2}a^2_{n/2}+\sum_{m>-n/2}a_{-m}a_{m+n}
  -(n+1)\rho a_n,\label{eq L a}
  \\ &C=1-12\rho^2,
  \end{align}
  where $\epsilon=0$ if $n$ is odd and $\epsilon=1$ if $n$ is even. 
  
 \begin{prop}
For $|\lambda\rangle\in\mathcal{F}_\lambda$, 
the operators $L_n$ ($n\ge r$) defined as \eqref{eq L a} act on $|\lambda\rangle$ as follows:
\begin{equation*}
L_n |\lambda\rangle=0\quad (n>2r),\quad 
L_n|\lambda\rangle=\left(\frac{1}{2}\sum_{m=0}^n\lambda_m\lambda_{n-m}-\delta_{n,r}(n+1)\rho\lambda_n\right)|\lambda\rangle\quad (n=r,r+1,\ldots,2r),
\end{equation*}
where $\lambda_i=0$ for $i<0$ or $i>r$. 
\end{prop}
The operators $L_n$ ($n\le -r$) defined as in \eqref{eq L a} act on $\langle\lambda|\in\mathcal{F}^*_\lambda$ similarly. 
As a result of this proposition, it is possible to define a bilinear pairing 
$\langle \cdot \rangle$: $V^{*,[s]}_{\Lambda'}\times V^{[r]}_\Lambda\to\C$.

\subsection{Vertex operators}  
  
  The free boson field $\varphi(z)$ is defined by 
  \begin{equation*}
  \varphi(z)=q+a_0\log(z)-\sum_{n\neq 0}\frac{a_n}{n}z^{-n},   
  \end{equation*}
  where $[a_m,q]=\delta_{m,0}$. 
  The commutation relations of the Heisenberg algebra imply the 
  operator product expansion at $|z|>|w|$ 
  \begin{equation*}
\varphi(z)\varphi(w)=\log(z-w)+:\varphi(z)\varphi(w):. 
\end{equation*}
Here, the normal order is defined by 
\begin{equation*}
:a_m a_n:=\left\{\begin{matrix}
a_m a_n\quad (n\ge 0),
\\
a_n a_m \quad (n<0),
\end{matrix}\right.
\quad :q a_n:  =\left\{\begin{matrix}
q a_n\quad (n\ge 0),
\\
a_n q \quad (n< 0). 
\end{matrix}\right.
\end{equation*}  
  
  Define $T(z)$ by $T(z)=:\varphi'(z)^2/2:+\rho\varphi^{(2)}(z)$. Then, the coefficients $L_n$ of $T(z)=\sum_{n\in\mathbb{Z}}L_nz^{-n-2}$ 
coincide with \eqref{eq L a}. 

The confluent primary field introduced in \cite{Nagoya Sun 2010} is defined 
  as 
\begin{equation*}
\Phi^{[r]}_\lambda(z)=:\exp\left(\sum_{n=0}^r \frac{\lambda_n}{n!}\varphi^{(n)}(z)\right): 
\end{equation*}
  for an ($r+1$)-tuple of parameters $\lambda=(\lambda_0,\ldots,
  \lambda_r)\in\C^r\times\C^\times$. 

\begin{prop}[\cite{Nagoya Sun 2010}]
If $|z|>|w|$, then we have 
\begin{align*}
T(z)\Phi^{[r]}_\lambda(w)=&\frac{1}{ 2}\left(\sum_{i=0}^r\frac{\lambda_i}{ (z-w)^{i+1}}\right)^2\Phi^{[r]}_\lambda(w)
+\frac{1}{ z-w}\partial_w\Phi^{[r]}_\lambda(w)
+\sum_{i=1}^r\frac{1}{ (z-w)^{i+1}}D_{i-1}\Phi^{[r]}_\lambda(w)
\\
&-\rho\sum_{i=0}^r\frac{(i+1)\lambda_i}{(z-w)^{i+2}}\Phi^{[r]}_\lambda(w)
+:T(z)\Phi^{[r]}_\lambda(w):,
\end{align*}
where $D_k=\sum_{p=1}^{r-k}p\lambda_{p+k}\partial/\partial \lambda_p$. 
\end{prop}
The above operator product expansion corresponds to the following 
  commutation relations. 
\begin{cor}
 \begin{align*}
\left[L_n,\Phi^{[r]}_\lambda(w)\right]=&w^{n+1}\partial_w\Phi^{[r]}_\lambda(w)
+\sum_{i=0}^{r-1}
\begin{pmatrix}
n+1
\\
i+1
\end{pmatrix}
w^{n-i}D_{i}\Phi^{[r]}_\lambda(w)
\\
&+\frac{1}{ 2}\sum_{i,j=0}^r \lambda_i\lambda_j\begin{pmatrix}n+1\\i+j+1\end{pmatrix}w^{n-i-j}\Phi^{[r]}_\lambda(w)
-\rho\sum_{i=0}^r\begin{pmatrix}
n+1
\\
i+1
\end{pmatrix}(i+1)\lambda_iw^{n-i}\Phi^{[r]}_\lambda(w). 
\end{align*}
\end{cor}  
  
\begin{prop}  
For $\lambda=(\lambda_0,\ldots,
  \lambda_r)\in\C^r\times\C^\times$ and 
  $\mu=(\mu_0,\ldots,
  \mu_s)\in\C^s\times\C^\times$, 
the confluent primary field $\Phi^{[r]}_\lambda(z)$ 
acts on $|\mu\rangle\in\mathcal{F}_\mu$ as follows. 
\begin{align*}
\Phi^{[r]}_\lambda(z)|\mu\rangle
=&z^{\lambda_0\mu_0}
\left(e^{\sum_{n=0}^r\lambda_na_{-n}/n}
e^{\sum_{m=0}^s\mu_ma_{-m}/m}+O(z)\right). 
\end{align*}
  \end{prop}
  Given these facts, we define general vertex operators as follows. 
  \begin{dfn}
  The rank $r$ vertex operator $\Phi^{[r],\lambda}_{\Delta,\Lambda}(z)$: 
$V^{[0]}_{\Delta}\to V^{[r]}_{\Lambda}$ with $\lambda=(\lambda_0,\ldots, \lambda_r)$, 
($\lambda_r\neq 0$),  
$\Lambda=(\Lambda_r,\ldots, \Lambda_{2r})$ ($\Lambda_{2r}\neq 0$) is  defined by
\begin{align}
\left[L_n,\Phi^{[r],\lambda}_{\Delta,\Lambda}(z)\right]=&z^{n+1}\partial_z\Phi^{[r],\lambda}_{\Delta,\Lambda}(z)
+\sum_{i=0}^{r-1}
\begin{pmatrix}
n+1
\\
i+1
\end{pmatrix}
z^{n-i}D_{i}\Phi^{[r],\lambda}_{\Delta,\Lambda}(z)\nonumber 
\\
&+\frac{1}{ 2}\sum_{i,j=0}^r \lambda_i\lambda_j\begin{pmatrix}n+1\\i+j+1\end{pmatrix}
z^{n-i-j}\Phi^{[r],\lambda}_{\Delta,\Lambda}(z)
-\rho\sum_{i=0}^r\begin{pmatrix}
n+1
\\
i+1
\end{pmatrix}(i+1)\lambda_iz^{n-i}\Phi^{[r],\lambda}_{\Delta,\Lambda}(z), \label{comrel_rankr}
\end{align}
where $D_k=\sum_{p=1}^{r-k}p\lambda_{p+k}\frac{\partial}{\partial \lambda_p}$
and 
\begin{equation}\label{eq_VO_act_0r}
\Phi^{[r],\lambda}_{\Delta,\Lambda}(z)|\Delta\rangle=z^\alpha \exp\left(\sum_{n=0}^r 
\frac{\beta_n}{z^n}\right)\sum_{m=0}^\infty v_mz^m,
\end{equation}
where $v_0=|\Lambda\rangle$, $v_m\in V^{[r]}_{\Lambda}$ ($m\ge 1$). 
  \end{dfn}
   \begin{dfn}
  For non-negative $r$, define a rank $0$ vertex operator 
$\Phi^{\Delta}_{\Lambda,\Lambda'}(z): 
V^{[r]}_{\Lambda}\to V^{[r]}_{\Lambda'}$  by 
\begin{align}
&[L_n, \Phi^{ \Delta}_{\Lambda,\Lambda'}(z)]=z^n \left(z\frac{\partial}{\partial z}+
(n+1)\Delta\right)\Phi^{ \Delta}_{\Lambda,\Lambda'}(z), \label{comrel_rank0}
\\
&\Phi^{\Delta}_{\Lambda,\Lambda'}(z)\Lam=z^\alpha \exp\left(\sum_{n=0}^r 
\frac{\beta_n}{z^n}\right)
\sum_{m=0}^\infty v_mz^m,\label{eq_VO_act_rr}
\end{align}
where $v_0=|\Lambda'\rangle$, $v_m\in V^{[r]}_{\Lambda'}$ ($m\ge 1$). 
  \end{dfn}
  When $r=0$, the rank $0$ vertex operator 
   is known to  
   exist uniquely if the Verma module $V_\Lambda'^{[0]}$ 
   is irreducible.  
The dual rank $r$ vertex operator $\Phi^{*,[r],\lambda}_{\Delta,\Lambda}(z)$: 
$V^{*,[0]}_{\Delta}\to V^{*,[r]}_{\Lambda}$ and 
the dual rank $0$ vertex operator $\Phi^{*,\Delta}_{\Lambda,\Lambda'}(z)$: 
$V^{*,[r]}_{\Lambda}\to V^{*,[r]}_{\Lambda'}$ are defined 
in the same manner. 
In the general case, it seems that a vertex operator $\Phi^{[p], \lambda}_{\Lambda,\Lambda'}(z)$: 
$V^{[r]}_{\Lambda}\to V^{[s]}_{\Lambda'}$  can be defined 
by the commutation relations \eqref{comrel_rankr} and 
\begin{equation}\label{act_VO_rs}
\Phi^{[p],\lambda}_{\Lambda,\Lambda'}(z)\Lam=z^\alpha \exp\left(\sum_{n=0}^p 
\frac{\beta_n}{z^n}\right)(|\Lambda'\rangle+O(z)), 
\end{equation}
where $p=\max(r,s)$. However, when $r=s=1$, it is observed 
that vertex operators satisfying \eqref{comrel_rankr} and \eqref{act_VO_rs} 
exist but are not unique. Hence, if we want to have uniqueness, 
we must add additional conditions.

  \begin{conj}\label{conj_VO_0r}
  The rank $r$ vertex operator $\Phi^{[r],\lambda}_{\Delta,\Lambda}(z)$: 
$V^{[0]}_{\Delta}\to V^{[r]}_{\Lambda}$ exists and is  uniquely determined by the given 
parameters $\Delta$,  $\lambda$, $a$ with $\beta_r=a \lambda_r$, $\alpha=\alpha(\lambda_0, a, \Delta)$, $\beta_i(\lambda_0,\ldots,\lambda_r, a, \Delta)$
and 
\begin{align*}
&\Lambda_n=\frac{1}{2}\sum_{i=0}^r \lambda_i\lambda_{n-i}+\delta_{n,r}\left((-1)^{r+1}r\beta_r-(r+1)\rho\lambda_r\right)\quad (n=r,\ldots, 2r),
\\
&D_i(\beta_k)=(-1)^i(k+i)\beta_{k+i}\quad (k=0,1,\ldots,r). 
\end{align*}
  \end{conj}
  We have checked Conjecture \ref{conj_VO_0r} for $r=1,2,3$, 
that is, in this case,  the first 
 $v_1$,\ldots, $v_{10}$  of \eqref{eq_VO_act_0r}  are 
uniquely determined. 
 
 For a partition 
$\lambda=(\lambda_1,\ldots,\lambda_n)$ ($\lambda_i\ge \lambda_{i+1}$), set  
$L_\lambda=L_{-\lambda_1+r}\cdots
L_{-\lambda_n+r}$. From 
Theorem 2.4 in \cite{FJK}, the vectors 
$L_\lambda|\Lambda'\rangle$ where $\lambda$ 
runs over the set of all partitions, 
form a basis of $V^{[r]}_{\Lambda'}$. 
Define the degree of $L_\lambda|\Lambda'\rangle$ by 
$\deg(L_\lambda|\Lambda'\rangle)=|\lambda|=\sum_{i=1}^n\lambda_n$ and for $v=\sum_{\lambda}
c_\lambda L_\lambda|\Lambda'\rangle$, 
$\deg(v)=\mathrm{max}\{\deg(L_\lambda|\Lambda'\rangle)\}$. Set  
\begin{equation*}
U_m=\left\{v\in V^{[r]}_{\Lambda'}\ \big|
\ 0\le \deg(v)\le m\right\}. 
\end{equation*}
A basis of $U_m$ consists of $L_\lambda|\Lambda'\rangle$ 
 ($0\le |\lambda|\le m$) and hence $\dim(U_m)=1+\sum_{i=1}^mp(i)$, where 
$p(i)$ is the number of all partitions of $i$. 
  \begin{thm}\label{thm_VO_rr}
  For any positive integer $r$, 
the rank $0$ vertex operator $\Phi^{ \Delta}_{\Lambda,\Lambda'}(z)$: 
$V^{[r]}_{\Lambda}\to V^{[r]}_{\Lambda'}$ exists 
and is  uniquely determined by the given 
parameters $\Lambda$, $\Delta$, $\beta_r$. 
In particular, 
\begin{equation*}
\Lambda'_n=\Lambda_n-\delta_{n,r}r\beta_r \quad (n=r,\ldots,2r),  
\end{equation*}
$\alpha$ and $\beta_n$ for $n=1,\ldots, r-1$ 
are  polynomials 
in $ \Delta$, $\beta_r$, $\Lambda_r$, 
\ldots, $\Lambda_{2r}$, $\Lambda_{2r}^{-1}$.  
Moreover,  $v_m\in U_m$ and 
the coefficients $c_\lambda$ of the vectors $
L_\lambda|\Lambda'\rangle$ in $v_m$ are 
uniquely determined as polynomials 
in $ \Delta$, $\beta_r$, $\Lambda_r$,\ldots, $\Lambda_{2r}$, 
$\Lambda_{2r}^{-1}$.  
\end{thm}
We see that $\Lambda_{2r}=0$ is the only singular point of the rank $0$ vertex operator. This is the reason we set $\Lambda_{2r}\neq 0$ in the definition of the rank $r$ 
Verma module. 
For the case of $\Lambda_{2r}=0$ and $\Lambda_{2r-1}\neq 0$, the universal Whittaker 
module is still irreducible \cite{LGZ}, 
 \cite{FJK}. We expect  
from observations that such an
irreducible Whittaker module describes 
irregular singularities of rank $r/2$.

The examples of vertex operators are provided 
in Appendix A. 
Let us examine the rank $r$ vertex operator  
and prove Theorem \ref{thm_VO_rr} in the next subsections. 
%
\subsection{On a rank $r$ vertex operator from a Verma module to an irregular Verma module}

 Notice that the commutation relations 
\eqref{comrel_rankr}
of the rank $r$ vertex operator  and the condition \eqref{eq_VO_act_0r} 
imply  
\begin{align}
\label{L_n action part 1}
L_n v_m=&\delta_{n,0}\Delta v_m
+(\alpha+m-n-(n+1)\rho\lambda_0+(n+1)D_0)v_{m-n}
\\
&+\sum_{i=1}^{r-1}\begin{pmatrix}
n+1
\\
i+1
\end{pmatrix}D_{i}v_{m-n+i}-
\sum_{i=1}^r\left(i\beta_i+(i+1)\rho\lambda_i\begin{pmatrix}
n+1
\\
i+1
\end{pmatrix}\right)v_{m-n+i}
\nonumber
\\
&+\sum_{i,j=1}^r
\begin{pmatrix}
n+1
\\
i
\end{pmatrix}
D_{i-1}(\beta_j)v_{m-n-1+i+j}
+\frac{1}{2}\sum_{i,j=0}^r \lambda_i\lambda_j\begin{pmatrix}n+1\\i+j+1\end{pmatrix}v_{m-n+i+j},  \nonumber
\end{align}
for any $n\ge 0$, 
where $v_{-n}=0$ for $n>0$, because  
\begin{equation*}
L_n\Phi^{[r],\lambda}_{\Delta,\Lambda}(z)|\Delta\rangle=\left[L_n,\Phi^{[r],\lambda}_{\Delta,\Lambda}(z))\right]|\Delta\rangle+\delta_{n,0}\Delta\Phi^{[r],\lambda}_{\Delta,\Lambda}(z)
|\Delta\rangle. 
\end{equation*}
In this subsection, suppose that \eqref{L_n action part 1} holds, that is, we assume 
the existence of a rank $r$ vertex operator 
satisfying \eqref{comrel_rankr} and 
\eqref{eq_VO_act_0r}. First, 
let us examine the case of $n=0$. Set $v_m=\tilde{v}_m\Lam$. 
\begin{prop}\label{prop n=0 rank r}
 $D_0(\beta_i)=i\beta_i$ ($i=1,\ldots, r$) and 
\begin{equation}\label{eq L_0 action}
[L_0,\tilde{v}_m]|\Lambda\rangle=\left(m\tilde{v}_m+D_0(\tilde{v}_m)\right)|\Lambda\rangle. 
\end{equation}
\end{prop}
\begin{proof}
From \eqref{L_n action part 1},  
\begin{equation}
L_0 \tilde{v}_m|\Lambda\rangle=(\alpha+m+D_0+\Delta-\rho\lambda_0)\tilde{v}_{m}|\Lambda\rangle
+\sum_{i=1}^r\left(D_{0}(\beta_i)-i\beta_i\right)\tilde{v}_{m+i}|\Lambda\rangle
+\frac{1}{2}\lambda_0^2\tilde{v}_{m}|\Lambda\rangle. 
\end{equation}
Hence, substituting $m=-r,\ldots,-1$ yields $D_0(\beta_i)=i\beta_i$ ($i=1,\ldots, r$).  
Consequently,  
\begin{equation}\label{eq L_0 action 0}
L_0 \tilde{v}_m|\Lambda\rangle=\left(\alpha+m+D_0+\Delta+\lambda_0\left(-\rho+\frac{1}{2}\lambda_0\right)\right)\tilde{v}_{m}|\Lambda\rangle. 
\end{equation}
In particular,  
\begin{equation}\label{eq L_0 action on highest weight vector}
L_0|\Lambda\rangle=\left(D_0+\alpha+\Delta+\lambda_0\left(-\rho+\frac{1}{2}\lambda_0\right)\right)|\Lambda\rangle.
\end{equation}
Then, by  \eqref{eq L_0 action on highest weight vector},  \eqref{eq L_0 action 0}  is rewritten as \eqref{eq L_0 action}.  
\end{proof}

\begin{prop}\label{prop 0<n<r rank r}
\begin{equation}\label{D_i b_k 1r}
D_i(\beta_k)=(-1)^i(k+i)\beta_{k+i}\quad (k=1,\ldots,r)
\end{equation}
for $i=0,\ldots,r-1$, where $\beta_k=0$ if $k>r$, and  for $n=1,\ldots,r-1$, 
\begin{align}
L_n\tilde{v}_m|\Lambda\rangle=&\sum_{k=1}^n\left(\begin{pmatrix}
n+1
\\
k+1
\end{pmatrix}\left(D_{k}-(k+1)\rho\lambda_k+\frac{1}{2}\sum_{i=0}^k\lambda_i\lambda_{k-i}\right)+\sum_{i=0}^{k}
\begin{pmatrix}
n+1
\\
i
\end{pmatrix}(-1)^{i-1}k\beta_k
\right)\tilde{v}_{m-n+k}|\Lambda\rangle\label{L_n v_m 0r}
\\
&+\left(\alpha+m-n+(n+1)D_0+(n+1)\lambda_0\left(-\rho+\frac{1}{2}\lambda_0\right)\right)\tilde{v}_{m-n}|\Lambda\rangle
\nonumber.
\end{align}
In particular, 
\begin{equation*}
L_n|\Lambda\rangle=\left(D_{n}-(n+1)\rho\lambda_n+\frac{1}{2}\sum_{i=0}^n\lambda_i\lambda_{n-i}+(-1)^{n+1}n\beta_n\right)|\Lambda\rangle. 
\end{equation*}
\end{prop}

\begin{proof}  
Substituting $m=-r,\ldots,n-1-r$ 
into \eqref{L_n action part 1} for 
$n=1,\ldots, r-1$ yields 
\begin{equation}\label{eq D zero}
\sum_{i=k-r}^{n+1}\begin{pmatrix}
n+1
\\
i
\end{pmatrix}D_{i-1}(\beta_{k-i})=0,
\end{equation}
for $k=r+2,\ldots, n+1+r$. Thus, substituting $m=n-r,\ldots, -1$ into \eqref{L_n action part 1} again yields  
\begin{equation}\label{eq D not zero}
\left(\sum_{i=1}^{n+1}\begin{pmatrix}
n+1
\\
i
\end{pmatrix}D_{i-1}(\beta_{k+1-i})-k\beta_k\right)=0
\end{equation}
for $k=n+1,\ldots, r$. 
From \eqref{eq D zero}, inductively,  
\begin{equation*}
D_{i-1}(\beta_{r-j})=0\quad (j=0,1,\ldots, i-2)
\end{equation*}
for $i=2,\ldots, r$. From \eqref{eq D not zero}, inductively,  
\begin{equation*}
D_i(\beta_k)=(-1)^i(k+i)\beta_{k+i}\quad (k=1,\ldots,r-i)
\end{equation*}
for $i=1,\ldots,r-1$. Hence, we obtain 
\eqref{D_i b_k 1r}, which yields  
  \eqref{L_n v_m 0r}. 
\end{proof}

In a similar manner, we have the following proposition. 
\begin{prop}\label{prop r-1<n<2r+1 rank r}
for $n=r,r+1,\ldots,2r$, 
\begin{equation}
L_n|\Lambda\rangle=\frac{1}{2}\sum_{i=0}^r \lambda_i\lambda_{n-i}
|\Lambda\rangle+\delta_{n,r}\left((-1)^{r+1}r\beta_r-(r+1)\rho\lambda_r\right)|\Lambda\rangle, 
\end{equation}
and for $n>2r$, 
\begin{equation}
L_n\tilde{v}_m|\Lambda\rangle=0\quad (0\le m<n-2r). 
\end{equation}
\end{prop}

Now, we consider the case when $\Delta=0$ and $V_0^{[0]}$ is 
the irreducible highest weight representation.  
\begin{prop}
A rank $r$ vertex operator $\Phi_{0,\Lambda}^{[r],\lambda}(z)$: 
$V^{[0]}_{0}\to V^{[r]}_{\Lambda}$ exists uniquely and acts on 
the highest weight vector $|0\rangle$ as 
\begin{equation*}
\Phi_{0,\Lambda}^{[r],\lambda}(z)|0\rangle 
=\sum_{m=0}^\infty \frac{1}{m!}L_{-1}^m|\Lambda\rangle
z^m,
\end{equation*}
where 
\begin{equation}\label{eq Lambda_n}
\Lambda_n=\frac{1}{2}\sum_{i=0}^r \lambda_i\lambda_{n-i}
-\delta_{n,r}(r+1)\rho\lambda_r \quad (n=r,r+1,\ldots, 2r). 
\end{equation}
Consequently, 
\begin{equation*}
\lim_{z\to 0}\Phi_{0,\Lambda}^{[r],\lambda}(z)|0\rangle
=|\Lambda\rangle. 
\end{equation*}
\end{prop}
\begin{proof}
From $L_{-1}|0\rangle=0$, 
\begin{align*}
L_{-1}\Phi_{0,\Lambda}^{[r],\lambda}(z)|0\rangle=&
\frac{\partial}{\partial z}\Phi_{0,\Lambda}^{[r],\lambda}(z)
|0\rangle
\\
=&\left(\frac{\alpha}{z}-\sum_{i=1}^r\frac{i\beta_i}{z^{i+1}}\right)
\Phi_{0,\Lambda}^{[r],\lambda}(z)
|0\rangle+z^\alpha\exp\left(\sum_{i=1}^r \frac{\beta_i}{z^i}\right)
\sum_{m=0}^\infty mv_mz^{m-1}.
\end{align*}
This relation is true if and only if $\alpha=\beta_1=\cdots=\beta_r=0$ 
and $v_m=L_{-1}^m/m!|\Lambda\rangle$. 
Hence, if the rank $r$ vertex operator $\Phi_{0,\Lambda}^{[r],\lambda}(z)$ 
exists, then it is unique. In addition, from Proposition 
\ref{prop r-1<n<2r+1 rank r}, the relations 
\eqref{eq Lambda_n} hold.  

We need to show the relations 
\eqref{L_n action part 1} for $n\ge 0$ are true, when $\alpha=\beta_1=\cdots=\beta_r=0$ 
and $v_m=L_{-1}^m/m!|\Lambda\rangle$. 
Proposition \ref{prop n=0 rank r} and \ref{prop 0<n<r rank r} 
imply the actions of $D_n$ ($n=0,1,\ldots,r-1$) on $|\Lambda\rangle$ 
is 
\begin{equation}\label{eq D_n}
D_n|\Lambda\rangle=\left(L_n+\frac{1}{2}\sum_{i=0}^n\lambda_i\lambda_{n-i}
-(n+1)\rho \lambda_n\right)|\Lambda\rangle. 
\end{equation}
This is achieved by realizing $|\Lambda\rangle$ as 
\begin{equation*}
|\Lambda\rangle=\exp\left(\lambda_0q+\sum_{n=1}^r\frac{\lambda_n}{n}a_{-n}\right). 
\end{equation*}
By \eqref{eq D_n}, the relations \eqref{L_n action part 1} read as 
\begin{equation*}
L_n\tilde{v_m}|\Lambda\rangle=\sum_{i=1}^{r-1}
\begin{pmatrix}
n+1
\\
i+1
\end{pmatrix}\tilde{v}_{m-n+i}L_i|\Lambda\rangle
+\tilde{v}_{m-n}\left( m-n+(n+1)L_0\right) |\Lambda\rangle,
\end{equation*}
which can be verified by straightforward computations. 
\end{proof}

 Therefore, if we consider zero or infinity 
 as an irregular singular point, then 
 we do not need the general rank $r$ vertex operator 
 $\Phi^{[r],\lambda}_{\Delta,\Lambda}(z)$: 
$V^{[0]}_{\Delta}\to V^{[r]}_{\Lambda}$. 

\subsection{On a rank $0$ vertex operator from a rank $r$ Verma module to a rank $r$ Verma module}
If the rank $0$ vertex operator 
$\Phi^{ \Delta}_{\Lambda,\Lambda'}(z)$: 
$V^{[r]}_{\Lambda}\to V^{[r]}_{\Lambda'}$ 
satisfies the commutation relations 
\eqref{comrel_rank0} and the condition 
\eqref{eq_VO_act_rr}, then 
for  $n\ge r$,  
\begin{equation}\label{L_n action rr}
L_{n}v_m=\sum_{i=0}^r\delta_{n,i+r}\Lambda_{i+r}v_m
-\sum_{i=1}^r i\beta_iv_{m+i-n} 
+\left(\alpha+(n+1)\Delta+m-n\right)v_{m-n}.
\end{equation}
Setting $m=0$ in \eqref{L_n action rr} yields  
\begin{equation*}
\Lambda_n'=\Lambda_n-\delta_{n,r}r\beta_r
\end{equation*}
for $n=r,r+1,\ldots,2r$.

Set 
\begin{equation*}
\widetilde{L}_n=L_n-\Lambda_n\quad (n=r,\ldots,2r),\quad 
\widetilde{L}_n=L_n\quad (n>2r). 
\end{equation*}
  Then, the relations \eqref{L_n action rr} are rewritten as  
\begin{equation}\label{L_n action rr1}
\widetilde{L}_{n+r}v_m=\delta_{n,0}r\beta_rv_m
-\sum_{i=1}^r i\beta_iv_{m+i-n-r} 
+\left(\alpha+(n+1)\Delta+m-n\right)v_{m-n-r}\quad (n\ge 0).
\end{equation}
Because of the commutation relations \eqref{comrel_rank0}, 
we need to show only the uniqueness and existence of $v_m$ ($m\ge 1$) 
such that \eqref{L_n action rr1} holds. 
Below, we prepare lemmas and corollaries needed 
to prove Theorem \ref{thm_VO_rr}.

\begin{lem}\label{lem1}
For any tuple $(\mu_1,\ldots,\mu_k)$ for $\mu_i\in\Z_{\ge 1}$, 
\begin{equation*}
L_{-\mu_1+r}\cdots L_{-\mu_k+r}|\Lambda'\rangle \in U_m,
\end{equation*}
where $m=\sum_{i=1}^k\mu_i$. 
\end{lem} 
\begin{proof}
Because $[L_{-i+r},L_{-j+r}]=(j-i)L_{-(i+j-r)+r}$,  moving $L_{-\lambda_i+r}$ to 
the right side, at least, reduces the degree 
by $\min(i+j, r)$.  
\end{proof}
 
\begin{lem}\label{lem2}
For a partition $\lambda$, 
\begin{align*}
&\deg\left(L_{n+r}L_\lambda|\Lambda'\rangle\right)
=|\lambda|\quad (n=0,1,\ldots,r),
\\
&\deg\left(L_{n+r}L_\lambda|\Lambda'\rangle\right) 
<|\lambda|\quad (n>r).
\end{align*}
\end{lem}
\begin{proof}
Because $[L_{n+r},L_{-i+r}]=
(n+i)L_{-(i-n-r)+r}$, the degree 
of $[L_{n+r},L_\lambda]|\Lambda'\rangle$ 
is less than $|\lambda|$.  
\end{proof} 
 
\begin{lem}\label{lem3}
For a partition $\lambda=(m^{k_m}, (m-1)^{k_{m-1}},\ldots, 2^{k_2},1^{k_1})$
($k_i\in\Z_{\ge 0}$) and $n\ge 1$,  
\begin{equation}\label{eq_LL_on_L1}
\widetilde{L}_{r+n}L_\lambda|\Lambda'\rangle=2n k_n\Lambda_{2r}
L_{\bar{\lambda}}|\Lambda'\rangle+v,
\end{equation}
where $\bar{\lambda}=(m^{k_m},\ldots,(n+1)^{k_{n+1}},n^{k_n-1},(n-1)^{k_{n-1}}, \ldots,1^{k_1})$ with $\deg(\bar{\lambda})=|\lambda|-n$ and $\deg(v)\le |\lambda|-n-1$. 
\end{lem}
\begin{proof}
Since the left hand side $\widetilde{L}_{r+n}L_\lambda|\Lambda'\rangle$ 
of \eqref{eq_LL_on_L1} is computed as 
\begin{equation*}
\widetilde{L}_{r+n}L_\lambda|\Lambda'\rangle=
\sum_{i=1}^mL_{-m+r}^{k_m}\cdots L_{-(i+1)+r}^{k_{i+1}}
\left[L_{r+n},L_{-i+r}^{k_i}\right]L_{-(i-1)+r}^{k_{i-1}}\cdots L_{-1+r}^{k_1}|\Lambda'\rangle,
\end{equation*}
we examine $\left[L_{r+n},L_{-i+r}\right]=(n+i)L_{-(i-r-n)+r}$. 

 Case 1: $1\le i-r-n$. From Lemma \ref{lem1}, the degree of 
 \begin{equation*}
 L_{-m+r}^{k_m}\cdots L_{-(i+1)+r}^{k_{i+1}}
\left[L_{r+n},L_{-i+r}^{k_i}\right]L_{-(i-1)+r}^{k_{i-1}}\cdots L_{-1+r}^{k_1}|\Lambda'\rangle
 \end{equation*}
 is $|\lambda|-r-n$.

 Case 2:  $-r\le i-r-n\le 0$.  From  
 $L_{-(i-r-n)+r}|\Lambda'\rangle=\Lambda'_{2r+n-i}|\Lambda'\rangle$ 
 and Lemma \ref{lem2}, 
 the degree of 
 \begin{equation*}
 L_{-m+r}^{k_m}\cdots L_{-(i+1)+r}^{k_{i+1}}
\left[L_{r+n},L_{-i+r}^{k_i}\right]L_{-(i-1)+r}\cdots L_{-1+r}^{k_1}|\Lambda'\rangle
 \end{equation*}
 is $|\lambda|-i$.

 Case 3: $i-r-n\le -r-1$.  
 From $L_{-(i-r-n)+r}|\Lambda'\rangle=L_{2r+n-i}|\Lambda'\rangle=0$ and Lemma \ref{lem2}, the degree of 
 \begin{equation*}
 L_{-m+r}^{k_m}\cdots L_{-(i+1)+r}^{k_{i+1}}
\left[L_{r+n},L_{-i+r}^{k_i}\right]L_{-(i-1)+r}^{k_{i-1}}\cdots L_{-1+r}^{k_1}|\Lambda'\rangle
 \end{equation*}
 is less than $|\lambda|-n-1$. 
 
 Thus, the degree is the highest, when $i=n$ in Case 2. 
 Therefore, \eqref{eq_LL_on_L1} holds. 
\end{proof}

The key lemma is the following. 

\begin{lem}\label{key_lemma}
For any $u\in U_m$ ($u=\sum_\lambda a_\lambda 
L_\lambda|\Lambda'\rangle$) and $n=1,\ldots,m$,  
\begin{equation*}
\widetilde{L}_{n+r}u=\sum_\lambda \Lambda_{2r}a_\lambda 
b_\lambda L_{\bar{\lambda}}|\Lambda'\rangle+v,
\end{equation*}
where the sum is  over all partitions $\lambda$ of $m$ that 
have a component $\lambda_i=n$ of $\lambda=(\lambda_1,\ldots, \lambda_k)$ 
and $\bar{\lambda}=(\lambda_1,\ldots, \lambda_{i-1},\lambda_{i+1},\ldots, \lambda_k)$, and $\deg(v)\le m-n-1$, $b_\lambda\in\Z_{\ge 1}$. 
\end{lem}
 \begin{proof}
 This follows immediately from Lemma \ref{lem3}. 
 \end{proof}

\begin{cor}\label{cor_n_kill}
For any positive integer $n$, if 
$\widetilde{L}_{n+r} u=0$ for  $u\in V^{[r]}_{\Lambda'}$, 
then $u\in U_0$. 
\end{cor} 
  
Let  a bilinear pairing 
$\langle \cdot \rangle$: $V^{*,[0]}_{\Delta^*}\times V^{[r]}_{\Lambda'}\to\C$ be 
defined as in Subsection 2.1 by the 
Heisenberg algebra and its Fock space. 
For a partition $\lambda=(\lambda_1,\ldots,
\lambda_k)$ ($\lambda_i\ge \lambda_{i+1}$), 
set $\widetilde{L}_\lambda
=\widetilde{L}_{\lambda_1+r}\cdots
\widetilde{L}_{\lambda_k+r}$. 

\begin{lem}\label{lem_pairing_LL}
For partitions $\lambda$ and $\mu$ 
such that $ |\lambda|\ge |\mu|$, 
\begin{equation*}
\langle \Delta^* |\widetilde{L}_\lambda
L_\mu|\Lambda'\rangle=\left\{\begin{matrix}
0 & (\lambda\neq \mu),
\\
\left(2\Lambda_{2r}\right)^{|\lambda|}
\prod_{i=1}^m i^{k_i}k_i! & 
(\lambda=\mu),
\end{matrix}\right.
\end{equation*}
where $\lambda=(m^{k_m},(m-1)^{k_{m-1}},
\ldots, 2^{k_2},1^{k_1})$ ($k_i\in\Z_{\ge 0}$). 
\end{lem}
\begin{proof}
From Lemma \ref{key_lemma}, 
 $\widetilde{L}_{\lambda_1+r}\cdots
\widetilde{L}_{\lambda_k+r}$ 
reduces the degree of $L_\mu$ 
to $|\mu|-|\lambda|$. 
Hence, if $ |\lambda|> |\mu|$, then 
$\langle \Delta^* |\widetilde{L}_\lambda
L_\mu|\Lambda'\rangle=0$. Suppose $|\lambda|=|\mu|$. Then,  
 from Lemma \ref{lem3} 
 $\widetilde{L}_{\lambda_i+r}$ reduces  the degree of 
 $L_{-j_1+r}\cdots L_{-j_k+r}|\Lambda'\rangle$ to 
 $j_1+\cdots+j_k-\lambda_i$ if there exists some $j_l$ such that $j_l=\lambda_i$,  
 and to less than $j_1+\cdots+j_k-\lambda_i$ if $j_l\neq \lambda_i$ for any $l$. 
 Hence, $\langle \Delta^* |\widetilde{L}_\lambda
L_\mu|\Lambda'\rangle\neq 0$ if and only if $\lambda=\mu$. 
The commutation relation $[L_{\lambda_i+r},L_{-\lambda_i+r}]=2\lambda_iL_{2r}$ yields the value 
of $\langle \Delta^* |\widetilde{L}_\lambda
L_\lambda|\Lambda'\rangle$. 
\end{proof}

Let a square matrix $(\langle \Delta^* |\widetilde{L}_\lambda
L_\mu|\Lambda'\rangle)_{n\le |\lambda|,|\mu|\le m}$ for $n,m\in\Z_{\ge 0}$ ($n<m$) 
be defined in such a way that  $\langle \Delta^* |\widetilde{L}_\lambda$, $L_\mu|\Lambda'\rangle$
are arranged in  order of increasing with respect to $|\lambda|$, $|\mu|$. 
For example, 
\begin{equation*}
(\langle \Delta^* |\widetilde{L}_\lambda
L_\mu|\Lambda'\rangle)_{1\le |\lambda|,|\mu|\le 2}
=\begin{pmatrix}
\langle \Delta^* |\widetilde{L}_{1+r}
L_{-1+r}|\Lambda'\rangle&\langle \Delta^* |\widetilde{L}_{1+r}
L_{-2+r}|\Lambda'\rangle&
\langle \Delta^* |\widetilde{L}_{1+r}
L_{-1+r}^2|\Lambda'\rangle
\\
\langle \Delta^* |\widetilde{L}_{2+r}
L_{-1+r}|\Lambda'\rangle&
\langle \Delta^* |\widetilde{L}_{2+r}
L_{-2+r}|\Lambda'\rangle
&
\langle \Delta^* |\widetilde{L}_{2+r}
L_{-1+r}^2|\Lambda'\rangle
\\
\langle \Delta^* |\widetilde{L}_{1+r}^2
L_{-1+r}|\Lambda'\rangle
&
\langle \Delta^* |\widetilde{L}_{1+r}^2
L_{-2+r}|\Lambda'\rangle
&
\langle \Delta^* |\widetilde{L}_{1+r}^2
L_{-1+r}^2|\Lambda'\rangle
\end{pmatrix}. 
\end{equation*}

\begin{cor}\label{cor_det}
For any  $n,m\in\Z_{\ge 0}$, 
\begin{equation*}
\det\left( (\langle \Delta^* |\widetilde{L}_\lambda
L_\mu|\Lambda'\rangle)_{n\le |\lambda|,|\mu|\le m} \right)
= 
\Lambda_{2r}^{\sum_{i=n}^mip(i)},
\end{equation*}
where  $p(i)$ is the 
partition number of $i$. 
\end{cor} 
 \begin{proof}
  Lemma \ref{lem_pairing_LL} implies that the matrix 
 $(\langle \Delta^* |\widetilde{L}_\lambda
L_\mu|\Lambda'\rangle)_{n\le |\lambda|,|\mu|\le m}$ 
is an upper  triangular matrix. Hence, its determinant 
is the product of all diagonal entries that are also computed 
in Lemma \ref{lem_pairing_LL}.  
 \end{proof}
 
\begin{lem}\label{lem_leftLL} 
For any tuple $(\lambda_1,\ldots,\lambda_k)$ such that $\lambda_i\in\Z_{\ge 1}$ and 
$\sum_{i=1}^k\lambda_i\le m$ ($m\in\Z_{\ge 1}$), 
 and $u\in U_m$,  
\begin{equation*}
\widetilde{L}_{\lambda_1+r}\cdots \widetilde{L}_{\lambda_k+r}u=
\sum_{\mu,\atop |\lambda|\le |\mu|\le m}a_\mu \widetilde{L}_\mu u. 
\end{equation*}
\end{lem}
 \begin{proof}
From the commutation relation 
$[L_{i+r},L_{j+r}]=(i-j)L_{(i+j+r)+r}$ for 
$i,j>0$,  
\begin{equation*}
\widetilde{L}_{\lambda_1+r}\cdots \widetilde{L}_{\lambda_k+r}=\sum_{\mu,\atop 
|\lambda|\le |\mu|}a_\mu \widetilde{L}_\mu.
\end{equation*}
Since $u\in U_m$, 
Lemma \ref{key_lemma} 
 implies 
$\widetilde{L}_\mu u=0$ for $|\mu|>m$. 
 \end{proof}

 \subsubsection{Proof of the existence of $v_m$}
 
 Let us construct $v_m$ in $U_m$. When $m=1$, 
 $v_1=c^{(1)}_{(1)}L_{-1+r}|\Lambda'\rangle+c^{(1)}_\phi|\Lambda'\rangle$. 
From \eqref{L_n action rr1}, 
\begin{align*}
&\widetilde{L}_rv_1=-(r-1)\beta_{r-1}|\Lambda'\rangle,
\\
&\widetilde{L}_{1+r}v_1=-r\beta_{r}|\Lambda'\rangle,
\\
&\widetilde{L}_{n+r}v_1=0\quad (n>1).
\end{align*}
Hence, $c^{(1)}_{(1)}$ and $\beta_{r-1}$ are solved as 
\begin{equation*}
c^{(1)}_{(1)}=-\frac{r\beta_r}{2\Lambda_{2r}},\quad 
\beta_{r-1}=\frac{r\beta_r\Lambda_{2r-1}}{2(r-1)\Lambda_{2r}}. 
\end{equation*}
Note that 
$v_1$ satisfies the relation \eqref{L_n action rr1}, 
even though  
$c^{(1)}_\phi$ is not determined.

Suppose that $v_i\in U_i$ ($1\le i\le k\le  r-1$) satisfy the relation 
\eqref{L_n action rr1},  and the coefficients 
$c_\lambda^{(i)}$
 in 
 $v_i=\sum_{\lambda}c_\lambda^{(i)}L_\lambda|\Lambda'\rangle$  
and $\beta_{r-i}$ 
for $i=1,\ldots,k$ 
are determined as   
polynomials in 
$\Delta$, $\beta_r$, 
$\Lambda_{r}$,\ldots, $\Lambda_{2r}$,  
 $c_\phi^{(1)}$,\ldots, 
$c_\phi^{(k-1)}$. Then, 
$v_{k+1}=\sum_{\lambda}c_\lambda^{(k+1)}L_\lambda|\Lambda'\rangle$ is constructed as follows. 

{\bf Step 1.} We compute $\langle \Delta^* |\widetilde{L}_\lambda v_{k+1}$
by 
\begin{equation*}
\left(\langle \Delta^* |\widetilde{L}_\lambda v_{k+1}\right)_{2\le |\lambda|\le k+1}
=\left(\langle \Delta^* |\widetilde{L}_\lambda
L_\mu|\Lambda'\rangle\right)_{2\le |\lambda|,|\mu|\le k+1} 
\cdot \left(c^{(k+1)}_\mu\right)_{2\le |\mu|\le k+1}. 
\end{equation*}
The left-hand side is expressed by $\Lambda_{r}$,\ldots, $\Lambda_{2r}$, 
$\Lambda_{2r}^{-1}$, $\beta_r$, 
 $c_\phi^{(1)}$,\ldots, 
$c_\phi^{(k-1)}$ from \eqref{L_n action rr1}. Because of Corollary \ref{cor_det}, 
the coefficients $c_\lambda^{(k+1)}$ ($2\le |\lambda|\le k+1$) are 
 uniquely solved  as polynomials in $\Lambda_{r}$,\ldots, $\Lambda_{2r}$, $\Lambda_{2r}^{-1}$, 
$\beta_r$, 
 $c_\phi^{(1)}$,\ldots, 
$c_\phi^{(k-1)}$. Note that 
the coefficients $c_\lambda^{(k+1)}$ ($2\le |\lambda|\le k+1$) do not depend on 
$\Delta^*$.  

{\bf Step 2.} 
We show that the relation \eqref{L_n action rr1} for $n=1$ is true except for the constant term
 and it is true for $n\ge 2$. 
 Denote the right-hand side 
of \eqref{L_n action rr1} by $X_{n,m}$. 
 From Step 1 and Lemma \ref{lem_leftLL}, 
\begin{equation}\label{eq_step2}
\langle \Delta^* |\widetilde{L}_\lambda\left( \widetilde{L}_{n+r}v_{k+1}-X_{n,k+1}\right)=0
\end{equation}
for $1\le |\lambda|\le k$ when $n=1$ and $0\le |\lambda|\le k+1-n$ when $n>1$. 
Since   
$ \widetilde{L}_{n+r}v_{k+1}-X_{n,k+1}$ is an element of $U_{k+1-n}$ as a result of \ref{key_lemma}, let 
$ \widetilde{L}_{n+r}v_{k+1}-X_{n,k+1}=\sum_{\mu}a^{(n)}_\mu L_\mu|\Lambda'\rangle$ 
where the sum is over all $\mu$ such that $0\le |\mu|\le k+1-n$. Then, from Lemma \ref{key_lemma}, 
we obtain 
\begin{equation*}
\left(\langle \Delta^* |\widetilde{L}_\lambda\left( \widetilde{L}_{n+r}v_{k+1}-X_{n,k+1}\right)
\right)_{\delta_{n,1}\le |\lambda|\le k+1-n}=
\left(\langle \Delta^* |\widetilde{L}_\lambda L_\mu|\Lambda'\rangle\right)
_{\delta_{n,1}\le |\lambda|,|\mu|\le k+1-n} 
\cdot \left(a^{(n)}_{\mu}\right)_{\delta_{n,1}\le |\mu|\le k+1-n}. 
\end{equation*}
Hence, by \eqref{eq_step2} and Corollary \ref{cor_det}, 
for $1\le |\mu|\le k$ when $n=1$ 
and $0\le |\mu|\le k+1-n$ when $n>1$,
$a^{(n)}_{\mu}$   are equal to zero. Therefore, the relation \eqref{L_n action rr1} for $n=1$ is true except for the constant term
 and  is also true for $n\ge 2$. 
  It is easy to see that 
 $\widetilde{L}_{n+r}v_{k+1}=0$ 
 for $n>k+1$ 
from Lemma \ref{key_lemma}, 
 because $v_{k+1}\in U_{k+1}$. 

{\bf Step 3.} We determine $c_{(1)}^{(k+1)}$ and $\beta_{r-k-1}$ by looking at 
the constant terms of the relation \eqref{L_n action rr1} for $n=0$ and $n=1$. 
They are equivalent to 
\begin{align*}
\Lambda_{2r-1}c^{(k+1)}_{(1)}+(r-k+1)\beta_{r-k-1}=Y_{0,k+1},
\\
2\Lambda_{2r}c^{(k+1)}_{(1)}=Y_{1,k+1},
\end{align*}
where $Y_{0,k+1}$ and $Y_{1,k+1}$ are polynomials in $\Lambda_{r}$,\ldots, $\Lambda_{2r}$, $\Lambda_{2r}^{-1}$, $\beta_r$, $\Delta$, 
 $c_\phi^{(1)}$,\ldots, 
$c_\phi^{(k)}$. 
Thus, since $\Lambda_{2r}\neq 0$, $c^{(k+1)}_{(1)}$ and $\beta_{r-k-1}$ are 
solved as polynomials in $\Lambda_{r}$,\ldots, $\Lambda_{2r}$, $\Lambda_{2r}^{-1}$, $\beta_r$, 
$\Delta$, 
 $c_\phi^{(1)}$,\ldots, 
$c_\phi^{(k)}$. Thus, we have proved that the relation \eqref{L_n action rr1}  is true for $n=1$.

{\bf Step 4.} Finally, we show that 
the relation \eqref{L_n action rr1} is true for $n=0$. 
Set $\widetilde{L}_{r}v_{k+1}-X_{0,k+1}=u\in U_k$.  
From 
the relation \eqref{L_n action rr1} for $n\ge 1$ 
which was proved in Steps 2 and 3,  
and
 $[\widetilde{L}_{n+r}, 
\widetilde{L}_r]=n\widetilde{L}_{n+2r}$ 
for $n\ge 1$, we obtain $\widetilde{L}_{n+r}u=0$ 
for any $n\ge 1$. Hence, from 
Corollary \ref{cor_n_kill}, $u$ is 
an element in $U_0$. In Step 3, we 
determined $c_{(1)}^{(k+1)}$ and $\beta_{r-k-1}$, 
hence the 
constant term of 
$\widetilde{L}_{n+r}v_{k+1}-X_{n,k+1}$ 
is equal to zero. 
Therefore, $u=0$. 

\medskip

In the same way, $v_m$ ($m\ge 1$) is uniquely constructed.   
When $m=r$, instead of $\beta_i$, 
the parameter $\alpha$ is  uniquely determined  and 
when $m>r$, $c_\phi^{(m-r)}$ is uniquely determined. 
Therefore, we have proved the existence 
of the rank $0$ vertex operator  
$\Phi^{ \Delta}_{\Lambda,\Lambda'}(z)$: 
$V^{[r]}_{\Lambda}\to V^{[r]}_{\Lambda'}$ 
such that 
\begin{equation*}
\Lambda'_n=\Lambda_n-\delta_{n,r}r\beta_r \quad (n=r,\ldots,2r), 
\end{equation*}
$v_m\in U_m$, and 
$\alpha$, $\beta_n$ for $n=1,\ldots, r-1$ 
and the coefficients $c_\lambda^{(m)}$ of the vectors $
L_\lambda|\Lambda'\rangle$ in $v_m$ are  
polynomials in 
$\Delta$, $\beta_r$, $\Lambda_r$,\ldots, 
$\Lambda_{2r}$, $\Lambda_{2r}^{-1}$. \qed

\subsubsection{Proof of the uniqueness of $v_m$}

In the proof of the existence of $v_m$, 
we proved that 
if we suppose $v_m\in U_m$, then 
$v_m$ is uniquely determined by 
the parameters $\Delta$, $\beta_r$, $\Lambda_r$,\ldots, 
$\Lambda_{2r}$. We need to show that 
if an element $v_m \in V^{[r]}_{\Lambda'}$ 
satisfies the relation \eqref{L_n action rr1} 
and elements $v_k\in U_k$ for $k=1,\ldots,m-1$ satisfy the relation \eqref{L_n action rr1}, 
then $v_m$ is also an element in $U_m$. 

Suppose $v_m\in U_k$ ($k>m$) and 
$v_i\in U_i$ for $i=1,\ldots,m-1$. 
From Lemma \ref{key_lemma}, for any 
positive integer $n$,  
\begin{equation*}
\widetilde{L}_{n+r}v_m=\sum_\lambda \Lambda_{2r}a_\lambda 
b_\lambda L_{\bar{\lambda}}|\Lambda'\rangle+v,
\end{equation*}
where the sum is over all partitions $\lambda$ of $k$ that 
have a component $\lambda_i=n$ of $\lambda=(\lambda_1,\ldots, \lambda_k)$ 
and $\bar{\lambda}=(\lambda_1,\ldots, \lambda_{i-1},\lambda_{i+1},\ldots, \lambda_k)$, and $\deg(v)\le k-n-1$, $b_\lambda\in\Z_{\ge 1}$.
On the other hand, from 
$\widetilde{L}_{n+r}v_m=X_{n,m}$ and the  assumption, 
$\widetilde{L}_{n+r}v_m\in U_{m-n}$. 
Hence, since $k>m$, we obtain $a_\lambda=0$ 
for $\lambda=(\lambda_1,\ldots, \lambda_k)$ 
such that $\lambda_i=n$ for some $i$. Therefore, 
$v_m\in U_{k-1}$. \qed

\section{Irregular conformal blocks}

\subsection{Definition}

At the present stage, we can define irregular conformal 
blocks with at most two irregular singular points. 
It is convenient to  consider zero and infinity as irregular singular points. Then,  an 
irregular conformal block is defined as
\begin{equation}
\Phi(\Delta,\Delta',\Lambda,\Lambda',z,w)=
\left(\langle \Lambda_{\infty}|\Phi^{*, \Delta'_{1}}_{\Lambda'_{\infty},\Lambda'_{1}}(w_{1})\circ \cdots \circ 
\Phi^{* \Delta'_{N}}_{\Lambda'_{N-1},\Lambda'_{N}}(w_N)\right)\cdot\left( 
\Phi^{ \Delta_M}_{\Lambda_{M-1},\Lambda_M}(z_M)\circ \cdots \circ 
\Phi^{ \Delta_{1}}_{\Lambda_{1},\Lambda_{0}}(z_{1})|\Lambda_0\rangle\right), 
\end{equation}
by the bilinear pairing $\langle \cdot \rangle$: $V^{*,[s]}_{\Lambda'_N}\times V^{[r]}_{\Lambda_M}\to\C$.
Here, $\Phi^{ \Delta'_{k}}_{\Lambda'_{k-1},\Lambda'_{k}}(w_{k})$ and 
$\Phi^{\Delta_k}_{\Lambda_{k-1},\Lambda_{k}}(z_{k})$ are rank $0$ vertex operators: 
\begin{equation*}
\Phi^{ \Delta'_{k}}_{\Lambda'_{k-1},\Lambda'_{k}}(w_{k}): 
\ V^{*,[s]}_{\Lambda'_{k-1}}\to V^{*,[s]}_{\Lambda'_k}\quad 
(k=1,\ldots,N),\quad 
\Phi^{\Delta_k}_{\Lambda_{k-1},\Lambda_{k}}(z_{k}): 
\ V^{[r]}_{\Lambda_{k-1}}\to V^{[r]}_{\Lambda_k}\quad 
(k=1,\ldots,M),  
\end{equation*}
with $\Lambda'_0=\Lambda'_\infty$.  
From Theorem \ref{thm_VO_rr}
 and the definition of the  bilinear pairing
 by the Fock spaces, 
the number of free parameters of 
the irregular conformal block 
$\Phi(\Delta,\Delta',\Lambda,\Lambda',z,w)$ 
is actually $2M+2N+r+s+1$. 

The irregular conformal block $\Phi(\Delta,\Delta',\Lambda,\Lambda',z,w)$ 
is expected to be a divergent series with respect to both the $z$ and $w$ variables if $s,r>0$. 
If $s=0$ and $r>0$, then $\Phi(\Delta,\Delta',\Lambda,\Lambda',z,w)$ is expected to be a 
convergent series with respect to the $w$ variables and a divergent series with  respect to the $z$ variables. 
 In addition, if $s=r=0$, then a regular conformal block
$\Phi(\Delta,\Delta',\Lambda,\Lambda',z,w)$ is believed to be absolutely convergent in the domain 
$|z_1|<\cdots<|z_M|<|w_1|<\cdots<|w_N|$.

\subsection{Null vector condition}

In the case of regular singularities, if one of the vertex operators 
satisfies   
the null vector condition
\begin{equation*}
\frac{\partial^2}{\partial z^2}\Phi^{\Delta_{2,1}}_{\Delta_1,\Delta_3}(z)+b^2 :T(z)
\Phi^{\Delta_{2,1}}_{\Delta_1,\Delta_3}(z):=0,
\end{equation*}
where $c=1+6(b+1/b)^2$ and $\Delta_{2,1}=-(3b^2+2)/4$, 
then the conformal block is a solution to the BPZ equation 
mentioned in Section 1. 
From the null vector condition, we know that if we set $\Delta_1=Q^2/4-P^2$ with $Q=b+1/b$, then 
$\Delta_3$ must be  $Q^2/4-(P\pm b/2)^2$. 

In the case of irregular singularities, 
we have  the following proposition. 
\begin{prop}
Let a rank $0$ vertex operator $\Phi^{ \Delta}_{\Lambda,\Lambda'}(z)$: 
$V^{[r]}_{\Lambda}\to V^{[r]}_{\Lambda'}$ 
such that 
\begin{equation*}\Phi^{ \Delta}_{\Lambda,\Lambda'}(z)\Lam=z^\alpha \exp\left(\sum_{i=1}^r\beta_i z^{-i}\right)\left(|\Lambda'\rangle+O(z)\right),
\end{equation*}
satisfy the null vector condition
\begin{equation}\label{eq-Null-Vect-cond}
\frac{\partial^2}{\partial z^2}\Phi^{\Delta}_{\Lambda,\Lambda'}(z)+b^2 :T(z)
\Phi^{\Delta}_{\Lambda,\Lambda'}(z):=0. 
\end{equation}
Then, 
\begin{align}
&-b^2\Lambda_{n}=\sum_{i=n-r}^ri(n-i)\beta_i\beta_{n-i}
\quad (r+1\le n\le 2 r),\label{eq_NullVec1}
\\
&-b^2 \Lambda_r=r\beta_r\left((r+1)(b^2+1)-2\alpha\right)
+\sum_{i=1}^{r-1}i(r-i)\beta_i\beta_{r-i}. 
\label{eq_NullVec2}
\end{align}
\end{prop}
\begin{proof}
A straightforward computation of 
\begin{equation*}
\frac{\partial^2}{\partial z^2}\Phi^{\Delta}_{\Lambda,\Lambda'}(z)
\Lam+  b^2:T(z)
\Phi^{\Delta}_{\Lambda,\Lambda'}(z):\Lam
\end{equation*}
yields the relations 
\eqref{eq_NullVec1} and 
\eqref{eq_NullVec2}. 
\end{proof}

Recall that due to Theorem \ref{thm_VO_rr}, 
$\alpha$ and $\beta_i$ ($i=1,\ldots,r-1$) 
are solved by $\Lambda_n$ ($n=r,\ldots,2r$),
 $\beta_r$ and $\Delta$. 
From the example, 
 one can observe that
the relation 
\eqref{eq_NullVec2} implies that 
the conformal dimension $\Delta$ is equal 
to $\Delta_{2,1}$.   
Therefore, together with the 
 relation 
 $-b^2\Lambda_{2r}=r^2\beta_r^2$,
   the null vector 
 condition of a rank $0$ vertex operator 
 provides the condition to 
 the two parameters $\Delta$ and $\beta_r$.

 Based on the null vector condition, 
 let us give an example of irregular conformal 
 blocks satisfying BPZ-type  differential equations. 
 In the next subsection, we 
 explain the irregular conformal block for the 
 case of the confluent hypergeometric 
 equation, that is, Kummer's equation. 
The differential equations 
satisfied by irregular conformal blocks 
have been  already presented in \cite{AFKMY 2010} and  \cite{GT}. 
Also refer to \cite{Nagoya Yamada}, where quantization of 
the  
Lax equations of the Painlev\'e equations 
were derived as partial differential 
systems satisfied by irregular conformal 
blocks.

\subsubsection{Kummer}\label{sec_Kummer}

The first example is the 
irregular conformal blocks having  
one irregular singular point $z_1$ and 
two regular singular points $z_2$, $z_3$ with one 
null vector condition.  
Keeping in mind that $z_1$, $z_3$ will be set to  
$0$, $\infty$, respectively, 
we express them as 
\begin{equation}\label{eq-BPZK-reg}
\left(\langle 0|
\Phi^{*,\theta_\infty^2}_{0,\theta_\infty^2}(z_3)
\Phi^{*,1/4}_{\theta_\infty^2,
(\theta_\infty\pm 1/2)^2}(z_2)\right)\cdot 
\left(
\Phi^{\lambda,[1]}_{0,\Lambda}(z_1)|0\rangle
\right)
\end{equation}
or 
\begin{equation}\label{eq-BPZK-irr}
\left(\langle 0|
\Phi^{*,\theta_\infty^2}_{0,\theta_\infty^2}(z_3)\right)\cdot
\left(\Phi^{1/4}_{\Lambda,\Lambda^\pm}(z_2)
\Phi^{\lambda,[1]}_{0,\Lambda}(z_1)|0\rangle
\right),  
\end{equation}
where we set the central charge $c=1$,  
$\lambda=(\lambda_0,\lambda_1)$, 
$\Lambda=(\Lambda_1,\Lambda_2)$, and 
$\Lambda^\pm=(\Lambda_1\pm \Lambda_2^{1/2},\Lambda_2)$.

From the null vector condition 
\eqref{eq-Null-Vect-cond}, and the 
commutation relations \eqref{comrel_rankr} and
\eqref{comrel_rank0}, 
these two irregular conformal blocks 
satisfy the confluent BPZ equation 
\begin{equation*}
\left(\frac{\partial^2}{\partial z_2^2}
-\frac{1}{2}\left(
\sum_{i=0}^1\frac{\lambda_i}{(z_2-z_1)^{i+1}}\right)^2-\frac{1}{z_2-z_1}
\frac{\partial}{\partial z_1}
-\frac{\lambda_1}{(z_2-z_1)^2}\frac{\partial}{\partial \lambda_1}
-\frac{1}{z_2-z_3}
\frac{\partial}{\partial z_3}
-\frac{\theta_\infty^2}{(z_2-z_3)^2}
\right)\Psi(z,\lambda)=0, 
\end{equation*}
 and from the 
$\mathfrak{sl}_2$ invariance 
$L_n|0\rangle=\langle 0 |L_n=0$ for $n=0,\pm 1$, 
they also satisfy 
\begin{align*}
&\sum_{i=1}^N\frac{\partial}{\partial z_i}\Psi(z,\lambda)=0,
\\
&\sum_{i=1}^N\left(z_i\frac{\partial}{\partial z_i}
+\sum_{p=1}^{r_i}p\lambda^{(i)}_p\frac{\partial}{\partial \lambda^{(i)}_p}
+\frac{\lambda^{(i)}_0}{2}(\lambda^{(i)}_0-2\rho)\right)\Psi(z,\lambda)=0,
\\
&\sum_{i=1}^N\left(z_i^2\frac{\partial}{\partial z_i}
+2z_i\sum_{p=1}^{r_i}p\lambda^{(i)}_p\frac{\partial}{\partial \lambda^{(i)}_p}
+\sum_{p=1}^{r_i-1}p\lambda^{(i)}_{p+1}\frac{\partial}{\partial \lambda^{(i)}_p}
+(\lambda^{(i)}_0z_i+\lambda^{(i)}_1)(\lambda^{(i)}_0-2\rho)\right)\Psi(z,\lambda)=0, 
\end{align*} 
where $N=3$, $r_1=1$, $r_2=r_3=0$, 
$\lambda_i^{(1)}=\lambda_i$ ($i=0,1$), $\lambda_i^{(2)}=0$ ($i=0,1$), 
$\lambda_0^{(3)}=\sqrt{2}\theta_\infty$, 
 $\lambda_1^{(3)}=0$, $\rho=0$. 

Together with these four equations,  
the two irregular conformal blocks 
become solutions to the ordinary 
differential equation with respect to $z_2$. 
We set $z_1=0$, $z_2=1/x$, 
$z_3=\infty$ and by scaling the variables $x$, 
we can set $\Lambda_2=1/4$ so that 
$\lambda_1=1/\sqrt{2}$. 
Furthermore, set 
\begin{equation*}
\lambda_0=\frac{2\alpha-\gamma}{\sqrt{2}},
\quad \theta_\infty=\frac{\gamma-1}{2}. 
\end{equation*}
Then, the ordinary differential equation
is transformed to Kummer's confluent 
hypergeometric equation
\begin{equation}\label{eq-Kummer}
\left(\frac{d^2}{d x^2}+
(\gamma-x)\frac{d}{dx}-\alpha\right)F(x,\alpha,\gamma)=0. 
\end{equation}
Here, $F(x,\alpha,\gamma)=g(x)\Psi(1/x,\alpha,\gamma)$ with $g(x)=x^{-\gamma/2}
\exp(x/2)$. 
Therefore, since  
\begin{equation*}
\lim_{z_3\to \infty}z_3^{2\theta_\infty^2}
\langle 0|
\Phi^{*,\theta_\infty^2}_{0,\theta_\infty^2}(z_3)=\langle \theta_\infty^2 |,\quad 
\lim_{z_1\to 0}
\Phi^{\lambda,[1]}_{0,\Lambda}(z_1)|0\rangle
=|(\Lambda_1,\Lambda_2)\rangle,
\end{equation*}
the irregular conformal blocks multiplied 
by the gauge factor $g(x)$: 
\begin{equation*}
g(x)\left(\langle \theta_\infty^2 |
\Phi^{*,1/4}_{\theta_\infty^2,
(\theta_\infty\pm 1/2)^2}(1/x)\right)\cdot 
|(\lambda_0/\sqrt{2},1/4)\rangle,\quad 
g(x)\langle \theta_\infty^2 |\cdot 
\left(\Phi^{1/4}_{(\lambda_0/\sqrt{2},1/4),(\lambda_0/\sqrt{2}\pm 1/2,1/4)}
|(\lambda_0/\sqrt{2},1/4)\rangle\right)
\end{equation*}
are solutions to Kummer's confluent 
hypergeometric equation. 

Let us examine these irregular conformal 
blocks.  From the definition, 
\begin{align*}
&\left(\langle \theta_\infty^2 |
\Phi^{*,1/4}_{\theta_\infty^2,
(\theta_\infty\pm 1/2)^2}(1/x)\right)\cdot 
|(\lambda_0/\sqrt{2},1/4)\rangle
=x^{\alpha'}\sum_{i=0}^\infty
A_ix^i,
\\
&\langle \theta_\infty^2 |\cdot 
\left(\Phi^{1/4}_{(\lambda_0/\sqrt{2},1/4),(\lambda_0/\sqrt{2}\pm 1/2,1/4)}
|(\lambda_0/\sqrt{2},1/4)\rangle\right)
=x^{\alpha} e^{\beta x}\sum_{i=0}^\infty
B_ix^{-i}, 
\end{align*} 
with $A_0=B_0=1$. 
Computing $\alpha'$, $\alpha$, $\beta$
and a few terms of $A_i$, $B_i$  
(see Appendix) yields 
\begin{align*}
&\left(\langle \theta_\infty^2 |
\Phi^{*,1/4}_{\theta_\infty^2,
(\theta_\infty\pm 1/2)^2}(1/x)\right)\cdot 
|(\lambda_0/\sqrt{2},1/4)\rangle=
x^{\pm\theta_\infty+1/2}e^{-x/2}
F\left(\begin{matrix}
\Lambda_1\pm\theta_\infty+\frac{1}{2}
\\
\pm 2\theta_\infty+1
\end{matrix};x\right),
\\
&\langle \theta_\infty^2 |\cdot 
\left(\Phi^{1/4}_{(\lambda_0/\sqrt{2},1/4),(\lambda_0/\sqrt{2}\pm 1/2,1/4)}
|(\lambda_0/\sqrt{2},1/4)\rangle\right)=
x^{\pm \Lambda_1}e^{\pm x/2}
F\left(\theta_\infty+\frac{1}{2}\mp\Lambda_1,
-\theta_\infty+\frac{1}{2}\mp\Lambda_1;\pm 
\frac{1}{x}\right),
\end{align*}
where 
\begin{align*}
F\left(\begin{matrix}
\alpha\\ \gamma
\end{matrix};x\right)=
\sum_{i=0}^\infty \frac{(\alpha)_i}
{i!(\gamma)_i}x^i,\quad 
F(\alpha,\gamma;x)=
\sum_{i=0}^\infty 
\frac{(\alpha)_i(\gamma)_i}{i!} x^i.
\end{align*}

\section{Painlev\'e tau functions}

In this section, 
as an application of the theory 
of irregular conformal blocks developed 
in Sections 2 and 3, 
we propose 
 series expansions of the tau 
 functions of the fourth and fifth Painlev\'e
 equations in terms of irregular 
 conformal blocks. 
Before proceeding to the detail, 
let us briefly review the Painlev\'e equations. 
 
 The relation between  the Painlev\'e functions and their tau functions 
is similar to that between the elliptic functions and  theta functions. 
Let us illustrate this by taking the first Painlev\'e equation as an example. 
The first Painlev\'e equation is 
\begin{equation*}
\frac{d^2y}{dt^2}=6y^2+t. 
\end{equation*}
The differential equation obtained by replacing the last term $t$ with 
the constant term is 
\begin{equation}\label{eq-2nd-elliptic}
\frac{d^2y}{dt^2}=6y^2-\frac{1}{2}g_2,  
\end{equation}
which is derived by differentiating the differential equation 
\begin{equation*}
\left(\frac{dy}{dt}\right)^2=4y^3-g_2y-g_3. 
\end{equation*}
Hence, the Weierstrass $\wp$ function is a solution to 
\eqref{eq-2nd-elliptic}. 
The Weierstrass $\sigma$ function is defined by 
\begin{equation*}
\wp=-\frac{d^2}{dt^2}\log \sigma. 
\end{equation*} 
The Weierstrass $\sigma$ function is one of the theta functions. 
Conversely, any elliptic function is expressed by a 
ratio of theta functions, 
which have explicit series expansions yielding various formulas involving the theta functions.

On the other hand, 
for any solution $y(t)$ of the first Painlev\'e equation, define the tau function $\tau_\mathrm{I}(t)$ by 
\begin{equation*}
y(t)=-\frac{d^2}{dt^2}\log \tau_\mathrm{I}(t). 
\end{equation*}
The Weierstrass $\wp$ function is an entire function as well as  $\tau_\mathrm{I}(t)$. 
The tau functions of the other Painlev\'e functions are defined in a similar manner, and 
they play an important role in the study of the Painlev\'e equations. 
The interested reader is referred  to \cite{Conte}, \cite{FIKN} 
for details and further references. 

Although the theta functions have explicit series expansions, explicit series 
expansions of the Painlev\'e tau 
functions were not known until recently. 
In 2012,  a remarkable discovery was reported by Gamayun, Iorgov and Lisovyy  \cite{GIL1}. They found 
that the sixth Painlev\'e tau function has a series expansion in terms of 
the four point conformal block: 
\begin{equation*}
\tau_{\mathrm{VI}}(t)=\sum_{n\in\mathbb{Z}}
s^n C\left(\begin{matrix}
\theta_1, \theta_t
\\
\theta_\infty,\sigma+n,\theta_0
\end{matrix}\right)\mathcal{F}\left(\begin{matrix}
\theta_1, \theta_t
\\
\theta_\infty,\sigma+n,\theta_0\end{matrix};t\right), 
\end{equation*}
where $s,\sigma\in\mathbb{C}$ are integral constants, $\theta_i$ are complex parameters 
in the sixth Painlev\'e equation, $\mathcal{F}(\theta,\sigma;t)$ is a 4-point conformal block  
with the central charge $c=1$: 
\begin{equation*}
\mathcal{F}\left(\begin{matrix}
\theta_1, \theta_t
\\
\theta_\infty,\sigma,\theta_0\end{matrix};t\right)=
\langle \theta_\infty^2|
\Phi^{*,\theta_1^2}_{\theta_\infty^2,\sigma^2}(1)\cdot 
\Phi^{\theta_t^2}_{\sigma^2,\theta_0^2}(t)
|\theta_0^2\rangle
\end{equation*}
and 
\begin{equation*}
C\left(\begin{matrix}
\theta_1, \theta_t
\\
\theta_\infty,\sigma,\theta_0
\end{matrix}\right)=\frac{\prod_{\epsilon,\epsilon'=\pm}
G(1+\theta_t+\epsilon \theta_0+\epsilon' \sigma)
G(1+\theta_1+\epsilon \theta_\infty
+\epsilon' \sigma)}{\prod_{\epsilon=\pm}G(1+2 \epsilon \sigma)},
\end{equation*} 
where $G(z)$ is the Barnes G-function such that $G(z+1)=\Gamma(z)G(z)$. By the AGT correspondence, 
\begin{align*}
\mathcal{F}\left(\begin{matrix}
\theta_1, \theta_t
\\
\theta_\infty,\sigma,\theta_0\end{matrix};t\right)=
t^{\sigma^2-\theta_0^2-\theta_t^2}(1-t)^{2\theta_0\theta_1}\sum_{\lambda, \mu\in \mathbb{Y}}\mathcal{F}_{\lambda,\mu}
\left(\begin{matrix}
\theta_1, \theta_t
\\
\theta_\infty,\sigma,\theta_0\end{matrix}\right)t^{|\lambda|+|\mu|},
\end{align*}
where $\mathbb{Y}$ stands for the set of all Young diagrams,  
\begin{align*}
\mathcal{F}_{\lambda,\mu}
\left(\begin{matrix}
\theta_1, \theta_t
\\
\theta_\infty,\sigma,\theta_0\end{matrix}\right)=&\prod_{(i,j)\in\lambda}
\frac{((\theta_t+\sigma+i-j)^2-\theta_0^2)((\theta_1+\sigma+i-j)^2-\theta_\infty^2)}
{h_\lambda^2(i,j)(\lambda_j'+\mu_i-i-j+1+2\sigma)^2}
\\
&\times \prod_{(i,j)\in\mu}
\frac{((\theta_t-\sigma+i-j)^2-\theta_0^2)((\theta_1-\sigma+i-j)^2-\theta_\infty^2)}
{h_\mu^2(i,j)(\mu_j'+\lambda-i-j+1-2\sigma)^2}. 
\end{align*}
Here, $\lambda=(\lambda_1,\ldots,\lambda_n)$ ($\lambda_i\ge \lambda_{i+1}$), 
$\lambda'$ denotes the transpose of $\lambda$, and $h_\lambda(i,j)$ is the hook length 
defined by $h_\lambda(i,j)=\lambda_i+\lambda_j'-i-j+1$. 
Therefore, we have an explicit series expansion of $\tau_\mathrm{VI}(t)$. 

A proof of the expansion of $\tau_{\mathrm{VI}}(t)$ was given in \cite{Iorgov Lisovyy Teschner} by constructing a fundamental solution to the linear problem 
of $\mathrm{P_{VI}}$, using Virasoro conformal field theory  and another proof was given 
in \cite{Bershtein Shchechkin} by proving that some conformal block 
satisfies the 
bilinear equations for $\mathrm{P_{VI}}$,
 using embedding of the direct sum of two Virasoro algebras in the sum of
fermion and super Virasoro algebra.

Later, series expansions of the tau functions in $t$, in other words, expansions of 
the tau functions 
at a {\it regular singular point}, 
 of the first line of the degeneration scheme
\begin{equation*}
\begin{diagram}
\node{\mathrm{P_{VI}}}\arrow{e}
\node{\mathrm{P_{V}}}\arrow{e}\arrow{se}
\node{\mathrm{P_{III}}}\arrow{e}\arrow{se}
\node{\mathrm{P_{III}^{D_7}}}\arrow{e}\arrow{se}
\node{\mathrm{P_{III}^{D_8}}}
\\
\node[3]{\mathrm{P_{IV}}}\arrow{e}
\node{\mathrm{P_{II}}}\arrow{e}
\node{\mathrm{P_I}}
\end{diagram}
\end{equation*}
were obtained in \cite{GIL2} by taking  scaling limits. 
It was conjectured in \cite{Its Lisovyy Tykhyy 2014} that a long-distance 
expansion of the tau function for $\mathrm{P_{III}^{D_8}}$, namely, an expansion in $t^{-1}$, 
can be represented as $\sum_{n\in\Z}s^n \mathcal{G}(\nu+n;t^{-1})$.  The first few terms of $\mathcal{G}(\nu;t^{-1})$ were explicitly obtained. 

Therefore, it is natural to expect that 
the tau functions of the Painlev\'e equations 
have series expansions in terms of conformal blocks. In the following, we present 
conjectural formulas of series expansions in $t^{-1}$ of the tau functions 
of the fifth and fourth Painlev\'e equations.

\subsection{Expansions of the $\mathrm{P_V}$ and $\mathrm{P_{IV}}$ tau functions}
The fourth and fifth Painlev\'e equations are the following 
second order nonlinear differential equations:
\begin{align*}
&\mathrm{P_{IV}}\qquad 
\frac{d^2 q}{dt^2}=\frac{1}{2q}\left(\frac{dq}{dt}\right)^2
+\frac{3}{2}q^3+4 tq^q+2(t^2-\alpha)q+\frac{\beta}{q},
\\
&\mathrm{P_V}\qquad 
\frac{d^2 q}{dt^2}=
\left(\frac{1}{2q}+\frac{1}{q-1}\right)\left(\frac{dq}{dt}\right)^2
-\frac{1}{t}\frac{dq}{dt}+\frac{(q-1)^2}{t^2}\left(\alpha q+\frac{\beta}{q}\right)
+\frac{\gamma q}{t}+\frac{\delta q(q+1)}{q-1},  
\end{align*}
$\alpha,\beta, \gamma,\delta$ being complex constants. They are equivalent 
to the Hamiltonian system:
\begin{equation*}
\frac{dq}{dt}=\frac{\partial H}{\partial p},\quad 
\frac{dp}{dt}=-\frac{\partial H}{\partial q}, 
\end{equation*}
with the Hamiltonians:
\begin{equation*}
H_\mathrm{IV}=2qp^2-(q^2+2tq-\theta-\theta_t)p
-\theta_t q,
\end{equation*}
where $\alpha=-(\theta+5\theta_t)/2$, 
$\beta=-(\theta+\theta_t)^2/2$, and 
\begin{equation*}
H_\mathrm{V}=(q-1)(qp-2\theta_t)(qp-p+2\theta)-tqp+((\theta+\theta_t)^2-\theta_0^2)q
+\left(\theta_t-\frac{\theta}{2}\right)t-\left(\theta_t+\frac{\theta}{2}\right),
\end{equation*}
where $\alpha=2\theta_0^2$, $\beta=-2\theta_t^2$, $\gamma=2\theta-1$, $\delta=-1/2$. 

For a solution ($q(t),p(t)$) to the Hamiltonian system, we define the Hamiltonian function
by 
\begin{equation*}
H_\mathrm{J}(t)=H(t;q(t),p(t))\quad (\mathrm{J=IV, V}).  
\end{equation*}
The $\tau$-functions $\tau_\mathrm{J}=\tau_\mathrm{J}(t)$ defined by 
\begin{equation*}
H_\mathrm{IV}(t)=\frac{d}{dt}\log \tau_\mathrm{IV}(t),\quad 
H_\mathrm{V}(t)=t\frac{d}{dt}\log \tau_\mathrm{V}(t) 
\end{equation*}
play a central role in the study of the Painlev\'e functions, such as 
the construction of B\"acklund transformations \cite{Okamoto}, 
relations to Soliton equations. 

A key to construct birational canonical transformations on the Painlev\'e functions is 
the nonlinear differential equations satisfied by the Hamiltonian functions 
\cite{Okamoto}. In fact,  
\begin{align}
&\left(H_\mathrm{IV}''\right)^2-(tH_\mathrm{IV}'-H_\mathrm{IV})^2+4H_\mathrm{IV}'(H_\mathrm{IV}'-\theta-\theta_t)(H_\mathrm{IV}'-2\theta_t)=0, 
\label{eq-Hamiltonian-DE-IV}
\\
&(tH_\mathrm{V}'')^2-(H_\mathrm{V}-tH_\mathrm{V}'+2(H_\mathrm{V}')^2)^2+\frac{1}{4}((2H_\mathrm{V}'-\theta)^2-4\theta_0^2)
((2H_\mathrm{V}'+\theta)^2-4\theta_t^2)=0,
\label{eq-Hamiltonian-DE-V} 
\end{align}
where $H_\mathrm{J}'=dH_\mathrm{J}/dt$ 
($\mathrm{J}=\mathrm{IV,V}$).
As mentioned above, the Hamiltonian functions are defined by a solution to the Painlev\'e equations. 
Inversely a function $q(t)$ defined by 
\begin{align*}
&\mathrm{P_{IV}}\quad q(t)=
\frac{H_\mathrm{IV}''-2tH_\mathrm{IV}'
+2H_\mathrm{IV}}{2(H_\mathrm{IV}'-2\theta_t)}, 
\\
&\mathrm{P_V} \quad q(t)=\frac{2(tH_\mathrm{V}''+H_\mathrm{V}-tH_\mathrm{V}'+2(H_\mathrm{V}')^2)}
{(2H_\mathrm{V}'-\theta)^2-4\theta_0^2}
\end{align*}
provides solutions to the fourth and fifth Painlev\'e equations, respectively. 

Based on the previous results on the series  expansions of the tau functions 
of the Painlev\'e equations 
$\mathrm{P_{VI}}$, 
$\mathrm{P_{V}}$, 
$\mathrm{P_{III}}$, 
$\mathrm{P_{III}^{D_7}}$ and 
$\mathrm{P_{III}^{D_8}}$ 
\cite{GIL1}, \cite{GIL2} and \cite{Its Lisovyy Tykhyy 2014}, 
we expect that the tau functions of the other cases  
also admit series expansions in terms of irregular conformal blocks. 
Let us recall that the building block of $\tau_{\mathrm{VI}}(t)$ 
is the 4-point regular conformal block 
with $c=1$:
\begin{equation*}
\langle \theta_\infty^2|
\Phi^{*,\theta_1^2}_{\theta_\infty^2,\sigma^2}(1)\cdot 
\Phi^{\theta_t^2}_{\sigma^2,\theta_0^2}(t)
|\theta_0^2\rangle. 
\end{equation*}
Thus, it is natural to expect that 
a building block of $\tau_{\mathrm{V}}(t)$ is 
the irregular conformal block having  
one irregular singular point of rank 1 and 
two regular singular points with $c=1$: 
\begin{equation*}
 \langle \theta_\infty^2 |\cdot \left(\Phi^{\theta_t^2}_{(\Lambda_1-\beta,\Lambda_2),(\Lambda_1,\Lambda_2)}
(t)|(\Lambda_1,\Lambda_2)\rangle\right), 
\quad
\left(\langle \theta_\infty^2 | \Phi^{*,\theta_t^2}_{\theta_\infty^2,\sigma^2}
(t)\right) \cdot|(\Lambda_1,\Lambda_2)\rangle. 
\end{equation*}
The latter is equal to the 
building block of the series 
expansion of the tau function 
of $\mathrm{P_V}$ obtained 
by a degeneration limit from 
the series expansion of $\tau_{\mathrm{VI}}(t)$ \cite{GIL2}.  Using  
the former irregular conformal block, 
we present a conjectural formula 
for the tau function of $\mathrm{P_V}$.  
In addition, it is natural to expect that a building block of $\tau_{\mathrm{IV}}(t)$ is 
the irregular conformal block having  
one irregular singular point of rank 2 and 
one regular singular point with $c=1$: 
\begin{equation*}
\langle 0| \cdot \left( \Phi^{\theta_t^2}_{(\Lambda_2-2\beta_2,\Lambda_3,\Lambda_4),(\Lambda_2,\Lambda_3,\Lambda_4)}(t)|(\Lambda_2,\Lambda_3,\Lambda_4)\right),  
\end{equation*}
where $\langle 0|\in \bar{V}^{*,[0]}_0$. 

After some computation,
we arrive at the following two 
conjectures. 

\begin{conj}[$\mathrm{P_{V}}$ case]
Let 
\begin{align*}
\tau(t)=\sum_{n\in\Z}&(-1)^{n(n+1)/2}s^n
\prod_{\epsilon=\pm 1}
G(1+\epsilon\theta_0+\theta-\beta-n)G(1+\theta_t+\epsilon(\beta+n))
\\
&\times \langle \theta_0^2|
\cdot \left( \Phi^{\theta_t^2}_{(\theta,1/4),(\theta-\beta-n,1/4)}
(1/t)|(\theta,1/4)\rangle\right) 
\end{align*}
and $H=t (\log(t^{-2\theta_t^2-\theta^2/2}e^{-\theta t/2}\tau(t)))'$. Then, 
$H$ satisfies the differential equation \eqref{eq-Hamiltonian-DE-V}. 
\end{conj}

\begin{conj}[$\mathrm{P_{IV}}$ case] 
Let
\begin{align*}
\tau(t)=t^{-2\theta_t^2}e^{\theta_t t^2/2}\sum_{n\in\Z}
s^n &G(1+\theta-\beta-n)
\prod_{\epsilon=\pm 1}
G(1+\theta_t+\epsilon(\beta+n))
\\
&\times \langle 0|\cdot \left( \Phi^{\theta_t^2}_{(\theta,0,1/4),(\theta-\beta-n, 0,1/4)}(1/t)
|(\theta,0,1/4)\rangle\right)
\end{align*}
and $ H=(\log \tau(t))'$. Then, $H$ satisfies the  
differential equation 
\eqref{eq-Hamiltonian-DE-IV}. 
\end{conj}

Here, 
\begin{align*}
&\langle \theta_0^2|
\cdot \left( \Phi^{\theta_t^2}_{(\theta,1/4),(\theta-\beta,1/4)}
(1/t)|(\theta,1/4)\rangle\right)
\\
&=t^{2\theta_t^2+2\beta(\theta-\beta)}
e^{\beta t}\left(1+2 \left(2 \beta^3-3 \beta^2 \theta +\beta  \theta^2-\beta  \theta_0^2-\beta  \theta_t^2+\theta  \theta_t^2\right)t^{-1} \right.
\\
&+2 \left(4 \beta ^6-12 \beta ^5 \theta +13 \beta ^4 \theta ^2-4 \beta ^4 \theta_0^2-4 \beta ^4 \theta_t^2+5 \beta ^4-6 \beta ^3 \theta ^3+6 \beta ^3 \theta  \theta_0^2+10 \beta ^3 \theta  \theta_t^2-10 \beta ^3 \theta +\beta ^2 \theta ^4
\right.
\\
&-2 \beta ^2 \theta ^2 \theta_0^2-8 \beta ^2 \theta ^2 \theta_t^2+6 \beta ^2 \theta ^2+\beta ^2 \theta_0^4+2 \beta ^2 \theta_0^2 \theta_t^2-3 \beta ^2 \theta_0^2+\beta ^2 \theta_t^4-3 \beta ^2 \theta_t^2+2 \beta  \theta ^3 \theta_t^2-\beta  \theta ^3
\\
&\left. \left.-2 \beta  \theta  \theta_0^2 \theta_t^2+\beta  \theta  \theta_0^2-2 \beta  \theta  \theta_t^4+5 \beta  \theta  \theta_t^2+\theta ^2 \theta_t^4-2 \theta ^2 \theta_t^2+\theta_0^2 \theta_t^2\right)t^{-2}+\cdots \right)
\end{align*}
and 
\begin{align*}
&\langle 0|\cdot \left( \Phi^{\theta_t^2}_{(\theta,0,1/4),(\theta-\beta, 0,1/4)}(1/t)
|(\theta,0,1/4)\rangle\right)
\\
&=t^{3\theta_t^2+\beta(2\theta-3\beta)}
e^{\beta t^2/2}\left(1+
\left(\theta ^2 \beta +2 \theta  \theta_t^2-6 \theta  \beta ^2-3 \theta_t^2 \beta +6 \beta ^3\right)t^{-2}
\right.
\\
&+\frac{1}{4} \left(2 \theta ^4 \beta ^2+8 \theta ^3 \theta_t^2 \beta -24 \theta ^3 \beta ^3-4 \theta^3 \beta +8 \theta ^2 \theta_t^4-60 
\theta^2 \theta_t^2 \beta^2-16 \theta^2 \theta_t^2+96 \theta^2 \beta^4\right. 
\\
&+48 \theta^2 \beta^2-24 \theta  \theta_t^4 \beta +120 \theta  \theta_t^2 \beta^3+72 \theta  \theta_t^2 \beta -144 \theta  \beta^5-140 \theta  \beta^3-2 \theta  \beta +18 \theta_t^4 \beta ^2
\\
&\left.\left.+\theta_t^4-72 \theta_t^2 \beta^4-66 \theta_t^2 \beta^2-\theta_t^2+72 \beta^6+105 \beta^4+3 \beta^2\right)t^{-4}+\cdots \right)
\end{align*}

If we substitute $H$ into \eqref{eq-Hamiltonian-DE-V} or \eqref{eq-Hamiltonian-DE-IV}, then
 the coefficient of $s^i$ for $i\in\Z$ 
 is of the form 
 \begin{equation*}
 t^A e^B(a_0+a_1t^{-1}+a_2t^{-2}+\cdots).
 \end{equation*}
Because $a_i$ ($i=0,1,\ldots,$) 
are finite sums of the coefficients 
of $t^{-k}$ in the corresponding irregular 
conformal blocks, we can check 
that the first several $a_i$'s are zero.

\appendix
\section{Data of irregular vertex operators}

\subsection{Vertex operators from a Verma module to an irregular Verma module}
Let $\lambda=(\lambda_0,\ldots, \lambda_r)$, 
($\lambda_r\neq 0$),  
$\Lambda=(\Lambda_r,\ldots, \Lambda_{2r})$ ($\Lambda_{2r}\neq 0$).  
We conjecture that 
 the rank $r$ vertex operator $\Phi^{[r],\lambda}_{\Delta,\Lambda}(z)$: 
$V^{[0]}_{\Delta}\to V^{[r]}_{\Lambda}$  
 such that 
\begin{equation*}
\Phi^{[r],\lambda}_{\Delta,\Lambda}(z)|\Delta\rangle=z^\alpha \exp\left(\sum_{n=1}^r 
\frac{\beta_n}{z^n}\right)\sum_{m=0}^\infty v_m|\Lambda\rangle z^m
\end{equation*}
exists uniquely, 
where $v_0=1$, $v_m|\Lambda\rangle\in V^{[r]}_{\Lambda}$ ($m\ge 1$), 
 $\beta_r=a \lambda_r$ ($a\in\C$), $\alpha=\alpha(\lambda_0, a, \Delta)$, $\beta_i(\lambda_0,\ldots,\lambda_r, a, \Delta)$ ($i=1,\ldots, r-1$), and
\begin{align*}
&\Lambda_n=\frac{1}{2}\sum_{i=0}^r \lambda_i\lambda_{n-i}+\delta_{n,r}(-1)^{r+1}r\beta_r\lambda_r\quad (n=r,\ldots, 2r),
\\
&D_i(\beta_k)=(-1)^i(k+i)\beta_{k+i}\quad (k=0,1,\ldots,r) 
\end{align*}
with $D_i=\sum_{j=1}^{r-i}j\lambda_{j+i}\partial/\partial \lambda_j$.
Below, we present $\alpha$, $\beta_i$ ($i=1,\ldots,r-1$), a few terms of $v_m$.

\subsubsection{Rank zero case}
Set $\Delta_1=\Delta$, $\Delta_2=\lambda_0^2/2-\lambda_0\rho$ and $\Delta_3=\Lambda_0$. 
\begin{align*}
\alpha=&\Delta_3-\Delta_2-\Delta_1,
\\
v_1=&\frac{(-\Delta_1+\Delta_2+\Delta_3)}{2 \Delta_3}L_{-1},
\\
v_2=&\frac{ c (\Delta_1-\Delta_2)^2-\Delta_3 \left(c \Delta_3+8 \Delta_3^2+12 \Delta_2\right)-(\Delta_1-\Delta_2-\Delta_3) \left(2 c \Delta_3+c+16 \Delta_3^2-4 \Delta_3\right)+8 \Delta_3 (\Delta_1-\Delta_2)^2}{4 \Delta_3 \left(2 c \Delta_3+c+16 \Delta_3^2-10 \Delta_3\right)}L_{-1}^2
\\
&-\frac{3 (\Delta_1-\Delta_2)^2-3 \Delta_3^2-(\Delta_1+\Delta_2-\Delta_3) (2 \Delta_3+1)}{2 c \Delta_3+c+16 \Delta_3^2-10 \Delta_3}L_{-2}.
\end{align*}

\subsubsection{Rank one case}

\begin{align*}
\alpha=&a^2-2 \Delta+2a \rho+a\lambda_0,
\\
v_1=&\frac{(a^2+\alpha)(\alpha+2\Delta)}{2a\lambda_1}-\frac{a }{\lambda_1}L_0+L_{-1}
\\
v_2=&\frac{a^2 }{2 \lambda_1^2}L_0^2-\frac{a}{\lambda_1}L_{-1}L_0+\frac{b_{21}}{\lambda_1^2}+\frac{1}{2} L_{-1}^2
+\frac{a^2+ (a^2+\alpha)(\alpha+2\Delta)}{2a\lambda_1}L_{-1}-\frac{a^2+(a^2+\alpha)(1+\alpha+2\Delta)}{2\lambda_1^2}L_0,
\\
v_3=&\frac{1}{6}L_{-1}^3-\frac{a^3}{5\lambda_1^3}L_0^3+\frac{a^2}{2\lambda_1^2}
L_{-1}L_0^2-\frac{a}{3\lambda_1}L_{-2}-\frac{a}{2\lambda_1}L_{-1}^2L_0
+\frac{a\left(3\left( a^2+ \alpha\right) (\alpha+2 \Delta +2)+2 \left(3 a^2+1\right)\right)}{12\lambda_1^3}L_0^2
\\
&+\frac{b_{32}}{\lambda_1^3}L_0
-\frac{\left(2a^2+(a^2+\alpha)(1+\alpha+2\Delta)\right)}{2\lambda_1^2}L_{-1}L_0+\frac{b_{34}}{\lambda_1^2}L_{-1}
+\frac{\left(2a^2+(a^2+\alpha)(\alpha+2\Delta)\right)}{4a\lambda_1}L_{-1}^2
+\frac{b_{36}}{\lambda_1^3},
\end{align*}
where $b_{ij}$ are polynomials in $a$, $\lambda_0$, $\Delta$, $\rho$. 
\subsubsection{Rank two case}
\begin{align*}
\alpha=& 6a^2-3\Delta-6a\rho-2a\lambda_0,
\quad
\beta_1=-2a\lambda_1,
\\
v_1=&L_{-1}+\frac{2a}{\lambda_2}L_1
+\frac{(\alpha+2\Delta_1)\lambda_1}
{\lambda_2},
\\
v_2=&\frac{1}{2}L_{-1}^2
+\frac{2a^2}{\lambda_2^2}L_1^2
+\frac{a\lambda_1(1+2\alpha+4\Delta)}{\lambda_2^2}L_1
+\frac{\lambda_1^2(\alpha+2\Delta)(1+\alpha+2\Delta)}
{2\lambda_2^2}
\\
&+\frac{-12 a^4+a^2 (12 \alpha-c+36 \Delta +1)+\alpha^2+2 \alpha\Delta-3 \Delta^2}{8 a \lambda_2 }+\frac{a}{\lambda_2}L_0
\\
&+\frac{2a}{\lambda_2}L_{-1}L_1
+\frac{\lambda_1(\alpha+2\Delta)}{\lambda_2}
L_{-1},
\\
v_3=&\frac{1}{6}L_{-1}^3
+\frac{4a^3}{3\lambda_2^3}L_1^3
+\frac{2a^2\lambda_1(1+\alpha+2\Delta)}{\lambda_2^2}L_1^2
-\frac{2a^2}{\lambda_2^2}L_0L_1
+\frac{2a^2}{\lambda_2^2}L_{-1}L_1^2
-\frac{a}{\lambda_2}L_{-1}L_0
+\frac{a}{\lambda_2}L_{-1}^2L_1
\\
&+\left(b_{31}\frac{\lambda_1^2}{\lambda_2^3}+b_{32}\frac{1}{\lambda_2^2}\right)L_1
-\frac{a\lambda_1(2+3\alpha+6\Delta)}{3\lambda_2^2}L_0
+\frac{a\lambda_1(1+2\alpha+4\Delta)}{\lambda_2^2}
L_{-1}L_1
\\
&
+b_{33}\frac{\lambda_1^3}{\lambda_2^3}
+b_{34}\frac{\lambda_1}{\lambda_2^2}
+\left(b_{35}\frac{\lambda_1^2}{\lambda_2^2}
+b_{36}\frac{1}{\lambda_2}\right)L_{-1}
+\frac{\lambda_1(\alpha+2\Delta)}{2\lambda_2}L_{-1}^2,
\end{align*}
where $b_{ij}$ are polynomials in $a$, $\lambda_0$, $\Delta$, $\rho$.

\subsubsection{Rank three case}
\begin{align*}
\alpha=&18a^2-4\Delta+12a\rho+3a\lambda_0,
\quad
\beta_1=3a\lambda_1,\quad 
\beta_2=-\frac{3}{2}a\lambda_2,
\\
v_1=&L_{-1}-\frac{3a}{\lambda_3}L_2
+\frac{3a\lambda_1^2
+(2\alpha-9a^2+6\Delta)\lambda_2}{2\lambda_3},
\\
v_2=&\frac{1}{2}L_{-1}^2
-\frac{3a}{\lambda_3}L_{-1}L_2
+\frac{3a}{2\lambda_3}L_1
+\frac{9a^2}{2\lambda_3^2}L_2^2
+\frac{12a\lambda_1^2-4(9a^2-2\alpha-6\Delta)\lambda_2}{8\lambda_3}L_{-1}
\\
&-\frac{36a^2\lambda_1^2+12a(1-9a^2+2\alpha+6\Delta)\lambda_2}{8\lambda_3^2}L_2
+\frac{b_1\lambda_1^4+b_2\lambda_1^2\lambda_2
+b_3\lambda_2^2}{\lambda_3^2}
+\frac{b_4\lambda_1}{\lambda_3},
\\
v_3=&-\frac{9 a^3 }{2 \lambda_3^3}L_2^3-\frac{9 a^2 }{2 \lambda_3^2}L_1L_2+\frac{9 a^2 }{2 \lambda_3^2}L_{-1}L_2^2+\frac{3 a }{2 \lambda_3}L_{-1}L_1-\frac{3 a }{2 \lambda_3}L_{-1}^2L_2+\frac{1}{6} L_{-1}^3-\frac{a}{\lambda_3}L_0
\\
&+\frac{9a^2(3a\lambda_1^2+(2-9a^2+2\alpha+6\Delta)\lambda_2)}{4\lambda_3^3}L_2^2
+\left(\frac{b_{5}\lambda_1^4+b_{6}\lambda_1^2\lambda_2+b_{7}\lambda_2^2}{\lambda_3^3}+\frac{b_{8}\lambda_1}{\lambda_3^2}\right)L_2
\\
&+\frac{3a(-3a\lambda_1^2+(-1+9a^2-2\alpha-6\Delta)\lambda_2)}{2\lambda_3^2}L_{-1}L_2
+\frac{a(9a\lambda_1^2+(4+27a^2+6\alpha+18\Delta)\lambda_2)}{4\lambda_3^2}L_1
\\
&+\left(\frac{b_{9}\lambda_1^4+b_{10}\lambda_1^2\lambda_2+b_{11}\lambda_2^2}{\lambda_3^2}
+\frac{b_{12}\lambda_1}{\lambda_3}\right)L_{-1}
+\frac{3a\lambda_1^2+(-9a^2+2\alpha+6\Delta)\lambda_2}{4\lambda_3}L_{-1}^2
\\
&+\frac{b_{13}\lambda_1^6
+b_{14}\lambda_1^4\lambda_2
+b_{15}\lambda_1^2\lambda_2^2
+b_{16}\lambda_2^3}{\lambda_3^3}
+\frac{b_{17}\lambda_1^3+b_{18}\lambda_1\lambda_2}{\lambda_3^2}
\end{align*}
\subsection{Vertex operators from an irregular Verma module to an irregular Verma module}

Let $\Lambda=(\Lambda_r,\ldots, \Lambda_{2r})$ ($\Lambda_{2r}\neq 0$). 
The rank zero vertex operator $\Phi^{\Delta}_{\Lambda,\Lambda'}(z)$: 
$V^{[r]}_{\Lambda}\to V^{[r]}_{\Lambda'}$  
is defined by the commutation relations 
\eqref{comrel_rank0} and 
\begin{equation*}
\Phi^{\Delta}_{\Lambda,\Lambda'}(z)|\Lambda\rangle=z^\alpha \exp\left(\sum_{n=1}^r 
\frac{\beta_n}{z^n}\right)\sum_{m=0}^\infty v_m|\Lambda'\rangle z^m,
\end{equation*}
where $v_0=1$, $v_m|\Lambda'\rangle\in V^{[r]}_{\Lambda'}$ ($m\ge 1$). 
Below, we present $\alpha$, $\beta_n$ ($n=1,\ldots,r-1$), a few terms of $v_m$.

\subsubsection{Rank one case}

\begin{align*}
\alpha=&-\frac{\beta_1(\Lambda_1-\beta_1)}{2 \Lambda_2}-2 \Delta,
\\
v_1=&-\frac{\beta_1}{2 \Lambda_2}L_0
+\frac{4\Delta\Lambda_1\Lambda_2-4\Delta\beta_1\Lambda_2
-3\beta_1^2\Lambda_1+2\beta_1^3}{8\Lambda_2^2},
\\
v_2=&\frac{\beta_1^2 }{8 \Lambda_2^2}L_0^2-\frac{\beta_1}{4 \Lambda_2}L_{-1}
-\frac{-8\Delta\Lambda_2^2
-2(1-2\Delta)\beta_1\Lambda_1\Lambda_2
+2(3-2\Delta)\beta_1^2\Lambda_2
+\beta_1^2\Lambda_1^2
-3\beta_1^3\Lambda_1+2\beta_1^4}{16\Lambda_2^3}L_0
\\
&+\left(48\Delta(\Delta-2)\Lambda_1^2\Lambda_2^2
+8(c-1+30\Delta-12\Delta^2)\beta_1\Lambda_1\Lambda_2^2
+12(2\Delta-1)\beta_1\Lambda_1^3\Lambda_2
\right.
\\
&-8(c-1+18\Delta-6\Delta^2)\beta_1^2\Lambda_2^2
-24(4\Delta-3)\beta_1^2\Lambda_1^2\Lambda_2
+3\beta_1^2\Lambda_1^4
+120(\Delta-1)\beta_1^3\Lambda_1\Lambda_2
\\
&\left.-18\beta_1^3\Lambda_1^3
-12(4\Delta-5)\beta_1^4\Lambda_2
+39\beta_1^4\Lambda_1^2
-36\beta_1^5\Lambda_1+12\beta_1^6\right)
\frac{1}{384\Lambda_2^4},
\\
v_3=&-\frac{\beta_1}{6\Lambda_2}L_{-2}
+\frac{\beta_1^2}{8\Lambda_2^2}L_{-1}L_0
-\frac{\beta_1^3}{48\Lambda_2^3}L_0^3
+\frac{1}{\Lambda_2^3}\sum_{i+j+2k=4}c^{(-1)}_{ijk} \beta_1^i\Lambda_1^j\Lambda_2^kL_{-1}
+\frac{1}{\Lambda_2^4}\sum_{i+j+2k=5}c^{(0,0)}_{ijk} \beta_1^i\Lambda_1^j\Lambda_2^kL_{0}^2
\\
&+\frac{1}{\Lambda_2^5}\sum_{i+j+2k=7}c^{(0)}_{ijk} \beta_1^i\Lambda_1^j\Lambda_2^kL_{0}
+\frac{1}{\Lambda_2^6}\sum_{i+j+2k=9}c_{ijk} \beta_1^i\Lambda_1^j\Lambda_2^k,
\end{align*}
where $c_{ijk}^{(l)}$ are polynomials 
in $c$, $\Delta$. 
\subsubsection{Rank two case}
\begin{align*}
\alpha=&\frac{\beta_2 \Lambda_3^2}{4 \Lambda_4^2}+\frac{3 \beta_2^2}{\Lambda_4}-\frac{\beta_2\Lambda_2}{\Lambda_4}-3\Delta,
\quad
\beta_1=\frac{\beta_2 \Lambda_3}{\Lambda_4},
\\
v_1=&b_{10}-\frac{\beta_2}{\Lambda_4}L_1,
\\
v_2=&\frac{\beta_2^2}{2
\Lambda_4^2}L_1^2-\frac{\beta_2}{2 \Lambda_4}L_0
+\frac{\beta_2}
{\Lambda_4}\left(\frac{\Lambda_3}{4\Lambda_4}
-b_{10}\right)L_1+b_{20},
\\
v_3=&\frac{\beta_2^2}{2 \Lambda_4^2}L_0L_1
-\frac{\beta_2}{3 \Lambda_4}L_{-1}
-\frac{\beta_2^3}{6 \Lambda_4^3}L_1^3
-\frac{\beta_2^2}{4 \Lambda_4^3}
\left(\Lambda_3-2 b_{10}\Lambda_4\right)L_1^2
+\frac{\beta_2}{6 \Lambda_4^2}
\left(\Lambda_3-3 b_{10}\Lambda_4\right)L_0
\\
&+\left(\frac{b_{10}\beta_2\Lambda_3}{4\Lambda_4^2 }
-\frac{\beta_2(2\Lambda_2-\beta_2)}{6 \Lambda_4^2}+\frac{b_{10}}{\Lambda_3}
-\frac{b_{20}\beta_2}{\Lambda_4}
\right)L_1+b_{30}. 
\end{align*}
Here, 
\begin{align*}
b_{10}=&\frac{\Lambda_3}{\Lambda_4}\left(\frac{\Delta }{2 }-\frac{\beta_2 \Lambda_3^2}{8 \Lambda_4^2}-\frac{3 \beta_2^2 }{2 \Lambda_4}+\frac{\beta_2 \Lambda_2}{2 \Lambda_4}\right) ,
\\
b_{20}=&\frac{1}{\Lambda_4^6}\sum_{2i+2j+3k+4l=22}b^{(20)}_{ijkl}
\beta_2^i\Lambda_2^j\Lambda_3^k\Lambda_4^l,
\quad 
b_{30}=\frac{1}{\Lambda_4^9}\sum_{2i+2j+3k+4l=33}b^{(30)}_{ijkl}
\beta_2^i\Lambda_2^j\Lambda_3^k\Lambda_4^l,
\end{align*}
where $b^{(20)}_{ijkl}$, $b^{(30)}_{ijkl}$ are 
polynomials in $c$, $\Delta$. 

\subsubsection{Rank three case}
\begin{align*}
\alpha=&\frac{3\beta_3\Lambda_4\Lambda_5}{\Lambda_6^2}-\frac{3\beta_3\Lambda_5^3}{16\Lambda_6^3}-\frac{3\beta_3\Lambda_3}{2\Lambda_6}-4\Delta,
\quad
\beta_1=\frac{3\beta_3\Lambda_4}{2\Lambda_6}
-\frac{3\beta_3\Lambda_5^2}{8\Lambda_6^2},
\quad
\beta_2=\frac{3\beta_3\Lambda_5}{4\Lambda_6},
\\
v_1=&b_{10}-\frac{3\beta_3}{2\Lambda_6}L_2,
\\
v_2=&b_{20}-\frac{3\beta_3}{4\Lambda_6}L_1
+\frac{9\beta_3^2}{8 \Lambda_6^2}L_2^2
-\frac{3 \beta_3}{8 
\Lambda_6^2}\left(
4 b_{10}\Lambda_6-\Lambda_5\right)L_2,
\\
v_3=&b_{30}-\frac{\beta_3}{2\Lambda_6}L_0
+\frac{9\beta_3^2}{8 
\Lambda_6^2}L_1L_2
+\frac{9\beta_3^2}{16 
\Lambda_6^3}\left(2 b_{10}\Lambda_6-\Lambda_5\right)L_2^2
-\frac{9\beta_3^3}{16 
\Lambda_6^3}L_2^3-\frac{ \beta_3}{4
\Lambda_6^2}\left(
3 b_{10}\Lambda_6-\Lambda_5\right)L_1
\\
&-\frac{\beta_3}{16 
\Lambda_6^3}\left(3 \Lambda_5 (\Lambda_5 - 2 b_{10} \Lambda_6)
+4 \Lambda_6( 6 b_{20} \Lambda_6-\Lambda_4)\right)L_2.
\end{align*}
Here, 
\begin{align*}
b_{10}=&\frac{15 \beta_3 \Lambda_5^4}{128 \Lambda_6^4}-\frac{9 \beta_3 \Lambda_4 \Lambda_5^2}{16 \Lambda_6^3}-\frac{27 \beta_3^2 \Lambda_5}{8 \Lambda_6^2}+\frac{3 \beta_3 \Lambda_4^2}{8 \Lambda_6^2}+\frac{\Delta  \Lambda_5}{2 \Lambda_6}+\frac{3 \beta_3\Lambda_3 \Lambda_5 }{4 \Lambda_6^2},
\\
b_{20}=&\frac{1}{\Lambda_6^8}
\sum_{3i+3j+4k+5l+6m=46}b_{ijklm}^{(20)}
\beta_3^i\Lambda_3^j\Lambda_4^k\Lambda_5^l
\Lambda_6^m,
\quad 
b_{30}=\frac{1}{\Lambda_6^{12}}
\sum_{3i+3j+4k+5l+6m=69}b_{ijklm}^{(30)}
\beta_3^i\Lambda_3^j\Lambda_4^k\Lambda_5^l
\Lambda_6^m,
\end{align*}
where $b^{(20)}_{ijklm}$, $b^{(30)}_{ijklm}$ are 
polynomials in $c$, $\Delta$. 


%

%
%

%

\bigskip

{\bf Acknowledgement.} 
The author is grateful to 
K.~Hiroe, M.~Jimbo, H.~Sakai and Y.~Yamada   
for suggestions and discussions.

\end{document}